\documentclass[review]{elsarticle}

\usepackage{lineno,hyperref}
\modulolinenumbers[5]

\usepackage{graphicx}

\usepackage[utf8x]{inputenc} 
\usepackage[T1]{fontenc}    
\usepackage{url}            
\usepackage{booktabs}       
\usepackage{amsfonts}       
\usepackage{nicefrac}       
\usepackage{lmodern}
\usepackage[bottom]{footmisc}

\usepackage{enumitem}
\usepackage[dvipsnames]{xcolor}

\usepackage{caption}
\usepackage{subcaption}
\usepackage{booktabs} 
\usepackage{threeparttable} 
\usepackage[justification=centering]{caption} 
\usepackage{longtable} 
\usepackage{microtype}      
\usepackage{lipsum}
\usepackage{amsmath}
\usepackage{amsthm}
\usepackage{mathtools} 
\usepackage{amssymb}
\usepackage{multirow}
\usepackage{multicol}
\usepackage[normalem]{ulem}
\usepackage[T1]{fontenc}
\DeclareMathOperator{\Tr}{Tr}

\newtheorem{proposition}{Proposition}

\newtheorem{lemma}{Lemma}
\newtheorem{remark}{Remark}

\usepackage{amsmath,amsfonts,bm}

\def\vtheta{{\bm{\theta}}}

\def\vk{{\bm{k}}}

\def\vm{{\bm{m}}}

\def\vr{{\bm{r}}}

\def\vx{{\bm{x}}}
\def\vy{{\bm{y}}}
\def\vz{{\bm{z}}}

\DeclareMathAlphabet{\mathsfit}{\encodingdefault}{\sfdefault}{m}{sl}
\SetMathAlphabet{\mathsfit}{bold}{\encodingdefault}{\sfdefault}{bx}{n}

\def\sN{{\mathbb{N}}}

\def\sR{{\mathbb{R}}}

\newcommand{\pr}[1]{\mathbb{P}\left(#1\right)}

\newcommand{\set}[1]{\left(#1 \right)}

\journal{Computational Statistics \& Data Analysis}

\biboptions{authoryear}

\begin{document}

\begin{frontmatter}

\title{Robust Prediction Interval estimation for Gaussian Processes by Cross-Validation method}

\author[cmap,total]{Naoufal Acharki \corref{cor1}}
\cortext[cor1]{Corresponding author: Naoufal Acharki \\ TotalEnergies OneTech, 91120 Palaiseau, France \\ Phone: +33 (0) 1 47 44 83 65 \\ Email: naoufal.acharki@polytechnique.edu}


\author[total]{Antoine Bertoncello}
\author[cmap]{Josselin Garnier}

\address[cmap]{TotalEnergies OneTech, 91120 Palaiseau, France}
\address[total]{Centre de Math\'ematiques Appliqu\'ees, Ecole Polytechnique, 91128 Palaiseau, France}

\begin{abstract}

{Probabilistic regression models typically use the Maximum Likelihood Estimation or Cross-Validation to fit parameters. These methods can give an advantage to the solutions that fit observations on average, but they do not pay attention to the coverage and the width of Prediction Intervals. A robust two-step approach is used to address the problem of adjusting and calibrating Prediction Intervals for Gaussian Processes Regression. First, the covariance hyperparameters are determined by a standard Cross-Validation or Maximum Likelihood Estimation method. A Leave-One-Out Coverage Probability is introduced as a metric to adjust the covariance hyperparameters and assess the optimal type II Coverage Probability to a nominal level. Then a relaxation method is applied to choose the hyperparameters that minimize the Wasserstein distance between the Gaussian distribution with the initial hyperparameters (obtained by Cross-Validation or Maximum Likelihood Estimation) and the proposed Gaussian distribution with the hyperparameters that achieve the desired Coverage Probability. The method gives Prediction Intervals with appropriate coverage probabilities and small widths.}

\end{abstract}


\end{frontmatter}

\nolinenumbers

\section{Introduction}
\label{sec1:Intro}

Many approaches of supervised learning focus on point prediction by producing a single value for a new point and do not provide information about how far those predictions may be from true response values. This may be inadmissible, especially for systems that require risk management. Indeed, an interval is crucial and offers valuable information that helps for better management than just predicting a single value.

The Prediction Intervals are well-known tools to provide more information by quantifying and representing the level of uncertainty associated with predictions. One existing and popular approach for prediction models without predictive distribution (e.g. Random Forest or Gradient Boosting models) is the bootstrap, starting from the Traditional bootstrap \citep{efron1994Bootstrap,Heskes97practicalconfidence} to Improved Bootstrap \citep{LI201897}. It is considered as one of the most used methods \citep{efron1994Bootstrap} for estimating empirical variances and for constructing Predictions Intervals, it is claimed to achieve good performance under some asymptotic framework.

A set of empirical methods have been proposed for these models to build Prediction Intervals like the Infinitesimal Jackknife \citep{Wager_confidenceintervals}, Jackknife-after-Bootstrap methods \citep{EfronJackknife_After_Bootstrap}, Quantile Random Forest \citep{QRF2006} Out-Of-Bag intervals \citep{Zhang_OOB_2018} and Conformal prediction \citep{Lei_SCIntervals2018, Candes_CI}. In the Deep Learning field, many recent methods have been also developed to quantify the uncertainty in Neural networks: The Delta method \citep{DeltaNNet97}, Mean-Variance Estimation \citep{Nix_MVENNET}, the Bayesian approach \citep{MacKay_BNnet1992,GhahramaniPMLR}, Lower Upper Bound Estimation \citep{LUBE_2011} and Quality-Driven ensembled approach \citep{Pearce18a}. Most methods estimate the Coverage Probability (CP) \citep{PICP_2008} and the mean Prediction Interval width (MPIW) \citep{KhosraviMPIW} by using the combinational Coverage Width-based Criterion (CWC) as a metric to identify model's parameters or define a loss function with a Lagrangian controlling the importance of the width and coverage. \cite{Pang2018} suggest the Receiver Operating Characteristic curve of Prediction Interval (ROC-PI), a graphic indicator that serves as a trade-off between the intervals width and CP for identifying the best parameters.

Unlike Ensemble methods or Neural Networks, there exist several prediction models with a probabilistic framework like the Gaussian Processes (GP) model \citep{Rasmussen2006} which are able to compute an efficient predictor with associated uncertainty. These models are more suitable for uncertainty quantification. They provide a predictive distribution with both point prediction and interval estimation and do not require any empirical approach such as the bootstrap. {In most cases, the predictive distributions of GP models are obtained either with a \textit{plug-in} method that takes the Maximum Likelihood Estimator (MLE) \citep{MardiaMarshall84,Stein_1999} of the model's hyperparameters or by using a Full-Bayesian approach that takes into account the posterior distribution of the hyperparameters and propagates it into the predictive distribution. However, both methods suffer from some limitations. Indeed, the MLE approach works well only when the model is well-specified and may fail in case of model misspecification \citep{bachoc2013}. At the same time, the Full-Bayesian approach is very complex to implement, typically with a Markov chain Monte Carlo (MCMC) algorithm and is sensitive to the choice of the \textit{prior} distribution of the hyperparameters \citep{FilipponeZG13,mure2018}. On the other hand, the calibration of Prediction Intervals is little studied in the literature. \cite{LawlessPI_2005} proposed a frequentist approach to predictive distribution to build and calibrate the Prediction Intervals. However, to the best of our knowledge, this approach has not yet been extended to models with a predictive distribution and we do not have any guarantees that it can work in the case of a misspecified model. Furthermore, improving the modelling of the covariance function seems to be efficient in overcoming the issue of a misspecified model. However, it may lead to complex covariance models and, consequently, severe difficulties in estimating the covariance function's hyperparameters, especially in high dimensions. Moreover, sometimes, it is challenging to find proper modelling without further knowledge of the system and the sources of uncertainty.} In this work, we propose a method based on Cross-Validation (CV) on the GP model to address the problem of model misspecification and calibrate Prediction Intervals by adjusting the upper and lower bounds to satisfy the desired level of CP. The method gives Prediction Intervals with appropriate coverage probabilities and small widths.

The paper is organized as follows. Section \ref{sec2:ProbForm} formulates the problem of Prediction Intervals estimation. Section \ref{sec3:GPRmodelling} introduces the Gaussian Process regression model and its training methods. In Section \ref{sec4:PIforGP}, we present a method for estimating robust Prediction Intervals supported by theoretical results. We show in Section \ref{sec5:Results} the application of this method to academic examples and to an industrial example. Finally, we present our conclusions in Section \ref{Conclusion}.

\section{Problem formulation}
\label{sec2:ProbForm}

We consider $n$ observations of an empirical model or computer code $f$. Each observation of the output corresponds to a $d$-dimensional input vector $\vx = \set{x_1,\ldots,x_d}^{\top} \in \mathcal{D} \subseteq \mathbb{R}^d$. The $n$ points corresponding to the model/code runs are called an experimental design $\mathbf{X} = \left( \vx^{(1)},\ldots,\vx^{(n)} \right)$ where $\vx^{(i)}= (x^{(i)}_1,\ldots,x^{(i)}_d)^{\top} \in \mathcal{D}$. The outputs are denoted by $\vy=\left(y^{(1)},\ldots,y^{(n)} \right) \in \sR^n$ with $y^{(i)}=f(\vx^{(i)})$. We seek to estimate the unobserved function $ \vx \in \mathcal{D} \mapsto f(\vx)$ from the data $\vy$ and make accurate predictions with the associated uncertainty.

Formally, let assume that $f$ is a realization of random process $Y$  and let $Y(\vx)$ be the value of model output at a point $\vx \in \mathcal{D}$, let $\alpha \in \left[0, 1\right]$ describes the nominal level of confidence.We wish to estimate the interval $\mathcal{PI}_{1-\alpha}$ with respect to the type II CP (the conditional Coverage Probability given the training set) such that the probability
\begin{equation}
\label{Eq:UncertaintyProb}
    \pr{Y(\vx) \in \mathcal{PI}_{1-\alpha}(\vx)\mid \mathbf{X}, \vy }
\end{equation}
is as close as possible to $1 - \alpha$. In most cases, $\mathcal{PI}_{1-\alpha}$ is a two-sided interval delimited by two bounds at $\vx \in \mathcal{D}$ 
\begin{equation}
\label{Eq:UncertaintyIC}
    \mathcal{PI}_{1-\alpha}\left( \vx \right) := \left[ y_{\alpha/2}(\vx) , \
    y_{1-\alpha/2}(\vx)\right] ,
\end{equation}
where $y_{\alpha/2}(\vx)= \tilde y(\vx) + z_{\alpha/2} \times \tilde \sigma(\vx)$ is the lower bound, $y_{1-\alpha/2}(\vx) = \tilde y(\vx) + z_{1-\alpha/2} \times \tilde \sigma(\vx)$ is the upper bound, $z_{\alpha/2}$ (resp. $z_{1-\alpha/2}$) is the $\alpha/2$ (resp. $1-\alpha/2$) quantile of the normalized predictive distribution (e.g. $t$-distribution for regression prediction), $\tilde y(\vx) = \mathbb{E}( Y(\vx) \mid \mathbf{X}, \vy )$ and $\tilde \sigma^2(\vx) = \mathrm{Var}(Y(\vx) \mid \mathbf{X}, \vy )$ are the predictive mean and variance.

In the framework of kriging, the prior distribution of the process $Y$ is Gaussian characterized by a mean and covariance. The Cumulative Distribution Function (CDF) of the predictive variable $Y(\vx)$ given $ \mathbf{X}$ and $\vy$ is well-defined and continuous with the Gaussian distribution. The quantile function is defined then as the inverse of the CDF and the quantiles $z_{\alpha/2}$ and $z_{1-\alpha/2}$ are fully characterized. Thus, estimating the interval $\mathcal{PI}_{1-\alpha}$ in equation (\ref{Eq:UncertaintyIC}) is equivalent to estimate the predictive mean $\tilde y(\vx) $ and variance $\tilde \sigma^2(\vx)$.

Therefore, the objective is to build a surrogate model to estimate correctly the upper and lower bounds of Prediction Intervals $\mathcal{PI}_{1-\alpha}$. This goes through the CV method with respect to the CP. In the following sections, $\| . \|$ refers to the Euclidean norm $\| . \|_2$ if applied to a vector and to the Frobenius norm, defined by $\| \mathbf{M} \|_F =  \left( \Tr \left(\mathbf{M} \mathbf{M}^{\top}\right) \right)^{1/2}$, if applied to a matrix.

\section{Modelling with Gaussian Processes}
\label{sec3:GPRmodelling}
 
We use the GP model to learn the unobserved function $f$. It is a Bayesian non-parametric regression (see \citet{Tipping2004} for Bayesian inference) which employs GP \textit{prior} over the regression functions. It will be converted into a \textit{posterior} over functions once some data has been observed. In the kriging framework \citep{Rasmussen2006,Stein_1999}, $Y$ is assumed \textit{a priori} to be a GP with mean $\mu(\vx)$ and covariance function $\vk(\vx, \vx') +\sigma_{\epsilon}^2 {\mathbf{1}}\{\vx=\vx'\}$ for all $\vx, \vx' \in \mathcal{D}$. $\sigma_{\epsilon}^2 \geq 0$ is the variance of measurement error, also called the nugget effect.

\subsection{The mean and covariance functions}
\label{sec:covfunctions}

The assumption made on the existing knowledge of the model $Y$ and the mean function $\mu$ defines three sub-cases of kriging
\begin{itemize}
	\item The Simple Kriging : $\mu$ is assumed to be known, usually null $\mu=0$.
	\item The Ordinary Kriging : $\mu$ is assumed to be constant but unknown.
	\item The Universal Kriging : $\mu$ is assumed to be of the form $\sum_{j=1}^{p} \beta_j f_{j-1}(x)$, where $f_j$ are predefined (e.g. polynomial functions $f_0(\vx)=1, f_j(\vx)=x_{j}, j=1,\ldots,p-1$) and unknown scalar coefficients $\beta_j$.
\end{itemize}

The covariance function $\vk$ is a map that is symmetric positive semi-definite, usually stationary $\vk(\vx, \vx') = \vr(\vx − \vx')$. The most commonly used kernel in $\mathbb{R}$ is the Mat\'ern kernel class {given by}
\begin{equation}
\label{eq:maternclass}
    \vr^{\nu}_{\sigma^2, \theta}(x − y)= \sigma^2 \frac{2^{1-\nu}}{\Gamma(\nu)}\left(\sqrt{2 \nu} \frac{|x-y|}{\theta }\right)^{\nu} K_{\nu}\left(\sqrt{2 \nu} \frac{|x-y|}{\theta }\right) , 
\end{equation}
{for $x,y \in \sR$}. Here $\sigma^2 > 0$ is the amplitude, $\theta > 0$ is the length-scale, $\Gamma $ is the complete Gamma function and $K_{\nu}$ is the modified Bessel function of the second kind. $(\sigma^2, \theta)$ are called hyperparameters. Some particular cases of Mat\'ern kernel are when $\nu=\frac{1}{2}$ (Exponential), $\nu=\frac{3}{2}$ (Mat\'ern 3/2), $\nu=\frac{5}{2}$ (Mat\'ern 5/2) and $\nu \rightarrow \infty $ (Gaussian or Squared-Exponential).

The choice of kernels is important in the kriging scheme and requires prior knowledge of the smoothness of the function $f$. For example, the choice of the Gaussian kernel assumes that the function is very smooth of class $\mathcal{C}^{\infty}$ (infinitely differentiable) which is often too strict as a condition. A common alternative is the functions Mat\'ern 5/2 or Mat\'ern 3/2 kernel

It is possible to build high-dimensional covariance models in $\mathbb{R}^d$ based on classical kernels in $\mathbb{R}$. In particular, the Matérn anisotropic geometric model (radial model), which we consider in the following in this paper, defined as

\begin{equation}
\label{radialModel}
    \vk^{\rm{radial}}_{\sigma^2, \vtheta}(\vx,\vx')=\vr^{\nu}_{\sigma^2, \vtheta} \left( \sqrt{ \sum_{j=1}^{d} \frac{|x_j-x_j'|^2}{\theta_j^2} } \right) ,
\end{equation}
where $\vr$ is a Mat\'ern kernel $\mathbb{R}$ as defined in (\ref{eq:maternclass}) and $\vtheta=(\theta_1,\ldots,\theta_d)$ the length-scale vector. The described method can be applied to other forms of covariance models like the tensorized product model with $d$-dimensional kernels ({as the product of kernels is also a kernel}) or the Power-Exponential model. In the following sections, instead of writing $\vk_{\sigma^2, \vtheta}$, we denote simply $\vk$ when there is no possible confusion.

\subsection{Gaussian Process Regression Model}

The \textit{prior} distribution of $Y$ on the learning experimental design $\mathbf{X}$ is multivariate  Gaussian
\begin{equation}
\label{eq:jointdist}
    \vy\mid \boldsymbol{\beta}, \sigma^2, \vtheta, \sigma^2_{\epsilon} \sim \mathcal{N}(\mathbf{F}  \boldsymbol{\beta}, \mathbf{K}) ,
\end{equation}

where
\begin{itemize}
    \item $\textbf{F}=(F_{i j}) \in \mathbb{R}^{n\times p}$ is the regression matrix such that $F_{i j} = f_j(\vx^{(i)})$.
    \item $\boldsymbol{\beta} = \{\beta_1,\ldots,\beta_p\} \in \sR^p$ are the regression coefficients.
    \item $\mathbf{K} = \left(\vk(\vx^{(i)},\vx^{(j)})\right)_{1\leq i,j\leq n} + \sigma^2_{\epsilon} \ \mathbf{I}_n \in \mathbb{R}^{n\times n}$ is the covariance matrix of the learning design $\mathbf{X}$.
\end{itemize}

{\it Hypothesis $\mathcal{H}_1$~: In the case of ordinary or universal kriging, we assume that $n \geq p$, $\mathbf{F}$ is a full rank matrix, and $\bf{e} \in $ {\rm Im }$\mathbf{F}$ where ${\bf e}=\left(1,\ldots,1\right)^{\top}$.}

\begin{remark}
    In Ordinary Kriging, the hypothesis $\mathcal{H}_1$ is always satisfied. In the Universal Kriging, the hypothesis $\bf{e} \in $ {\rm Im }$\mathbf{F}$  is satisfied as soon as the constant function $f_0(\vx)= C$ is included in the chosen family of functions $f_i$.
\end{remark}

\subsection{Prediction}

The Gaussian conditioning theorem is useful to deduce the \textit{posterior} distribution. By considering a new point $\vx_{\rm{new}}$, it can be shown that the predictive distribution of $Y(\vx_{\rm{new}})$ conditioned on the learning sample $\mathbf{X},\vy$ is also Gaussian

\begin{equation}
\label{eq:posteriorlaw}
    Y(\vx_{\rm{new}})\mid \mathbf{X},\vy, \sigma^2, \vtheta, \sigma^2_{\epsilon} \sim \mathcal{N}\left(\tilde y(\vx_{\rm{new}}), \tilde \sigma^2(\vx_{\rm{new}})\right) ,
\end{equation}

where, in the case of Ordinary or Universal Kriging {and by denoting $f_{\rm{trend}}(\vx_{\rm{new}}) = \left( f_j(\vx_{\rm{new}}) \right)^{p-1}_{j=0}$}, $y(\vx_{\rm{new}})$ and $\tilde \sigma^2(\vx_{\rm{new}})$ are given by the Best Linear Unbiased Predictor (BLUP),
\begin{equation}
    \label{eq:GPMean}
    \tilde y_{\sigma^2, \vtheta, \sigma_{\epsilon}^2 } (\vx_{\rm{new}}) = f_{\rm{trend}}(\vx_{\rm{new}})^{\top} \widehat{\boldsymbol{\beta}}  + \vk(\vx_{\rm{new}},\mathbf{X})^{\top}\mathbf{K}^{-1}(\vy-\mathbf{F} \widehat{\boldsymbol{\beta}} ) ,
\end{equation}
\begin{equation}
    \label{eq:GPVariance}
    \begin{aligned}
        \tilde \sigma_{\sigma^2, \vtheta, \sigma_{\epsilon}^2}^2&(\vx_{\rm{new}}) = \vk(\vx_{\rm{new}},\vx_{\rm{new}}) + \sigma_{\epsilon}^2 - \vk(\vx_{\rm{new}},\mathbf{X})^{\top}\mathbf{K}^{-1} \ \vk(\vx_{\rm{new}},\mathbf{X}) + \left( f_{\rm{trend}}(\vx_{\rm{new}}) - \right. \\ &\left. \mathbf{F} \mathbf{K}^{-1} \vk(\vx_{\rm{new}},\mathbf{X})  \right)^{\top} \left( \mathbf{F}^{\top} \mathbf{K}^{-1} \mathbf{F} \right)^{-1} \left( f_{\rm{trend}}(\vx_{\rm{new}}) -\mathbf{F} \ \mathbf{K}^{-1} \vk(\vx_{\rm{new}},\mathbf{X})  \right) ,
    \end{aligned}
\end{equation}

and,
\begin{equation}
\label{eq:beta_opt}
    \begin{aligned} &\widehat{\boldsymbol{\beta}} = \left(\mathbf{F}^{\top} \mathbf{K}^{-1} \mathbf{F}\right)^{-1} \mathbf{F}^{\top} \mathbf{K}^{-1} \vy. 
    \end{aligned}
\end{equation}

{We refer to \citet{Santner_2003} in Chapter 4 for a detailed proof of Equations (\ref{eq:posteriorlaw}-\ref{eq:beta_opt}). In particular, we note that the additional term of the predictive variance in (\ref{eq:GPVariance}) is due to the propagation of the non-informative improper form of prior distribution on the estimation of $\boldsymbol{\beta}$. }

{In the following and when there is no possible confusion,} $\tilde y_{\sigma^2, \vtheta, \sigma_{\epsilon}^2}$ (resp. $\tilde \sigma^2_{\sigma^2, \vtheta, \sigma_{\epsilon}^2}$) will be also denoted by $\tilde y$ (resp. $\tilde \sigma^2$) without specifying its dependence on hyperparameters or the nugget effect.

The most outstanding advantage of GP model compared to other models relies on the previous equations (\ref{eq:GPMean}) and (\ref{eq:GPVariance}). The predictive distribution can be used for sensitivity analysis \citep{OHagan2004ProbabilisticSA} and uncertainty quantification instead of costly methods based on Monte Carlo algorithms. Other possible considerations and extensions of GP modelling are described in \citep{Currin1991BayesianPO,Rasmussen2006}.

Given a GP regression model and a point $\vx_{\rm{new}} \in \mathcal{D}$, the \textit{posterior} {distribution} of prediction in  (\ref{eq:posteriorlaw}) can be standardized into

\begin{equation}
\label{eq:standardlaw}
    \tilde Z(\vx_{\rm{new}}) = \frac{ Y(\vx_{\rm{new}}) - \tilde y(\vx_{\rm{new}})}{\tilde \sigma(\vx_{\rm{new}})} \ \big| \ \mathbf{X},\vy, \sigma^2, \vtheta, \sigma^2_{\epsilon} \sim \mathcal{N}\left(0, 1\right) .
\end{equation}

By considering the standardized variable $\tilde Z(\vx_{\rm{new}})$, the $\alpha$-quantiles $z_{\alpha}$ are those of the standard normal {distribution} : $q_{1-\alpha/2}=\mathbf{\Phi}^{-1}(1-\alpha/2)$ and $q_{\alpha/2}=\mathbf{\Phi}^{-1}(\alpha/2)=-q_{1-\alpha/2}$ where $\mathbf{\Phi}$ is the CDF of the standard normal distribution, such that the Prediction Intervals $\mathcal{PI}_{1-\alpha}$ in (\ref{Eq:UncertaintyIC}) can be written as

\begin{equation}
\label{eq:UncertaintyIC_GP}
    \mathcal{PI}_{1-\alpha}\left(\vx_{\rm{new}} \right) = \left[ \tilde y(\vx_{\rm{new}}) - q_{1-\alpha/2} \times \tilde \sigma(\vx_{\rm{new}}) ; \
    \tilde y(\vx_{\rm{new}}) + q_{1-\alpha/2} \times \tilde \sigma(\vx_{\rm{new}})   \right] ,
\end{equation}

which gives a natural definition for $y_{\alpha/2}$ and $y_{1-\alpha/2}$ 
\begin{equation}
    y_{\alpha/2}\left(\vx \right) =  \tilde y(\vx) - q_{1-\alpha/2} \times \tilde \sigma(\vx)  \ ; \  y_{1-\alpha/2}\left(\vx \right) =  \tilde y(\vx) + q_{1-\alpha/2} \times \tilde \sigma(\vx) .
\end{equation}

\subsection{Training model with Maximum Likelihood method}
\label{par:MLE}

Constructing a GP model and computing the kriging mean and variance as shown in (\ref{eq:GPMean}) and (\ref{eq:GPVariance}) implies estimating the nugget effect $\sigma^2_{\epsilon}$ and the covariance parameters $(\sigma^2,\vtheta)$. Here, we assume that $\sigma^2_{\epsilon}$ is known or has been estimated by the method proposed in \cite{Iooss2017} for instance.

The Maximum Likelihood Estimator (MLE) $\hat{\sigma}_{ML}^2$ and $\hat{\vtheta}_{ML}$ of  $\sigma^2$ and $\vtheta$ is given by \citep{Santner_2003} 
\begin{equation}
\label{eq:OptimML}
    (\hat{\sigma}_{ML}^2, \hat{\vtheta}_{ML}) \in \operatorname{argmin}_{\sigma^2, \vtheta} \ \vy^{\top} \left( \mathbf{K}^{-1} - \mathbf{K}^{-1} \mathbf{F} \left( \mathbf{F}^{\top} \mathbf{K}^{-1} \mathbf{F} \right)^{-1} \mathbf{F}^{\top} \mathbf{K}^{-1} \right) \vy + \log \left( \det \mathbf{K} \right).
\end{equation}

The MLE method is optimal when the covariance function is well-specified \citep{Bachoc2013CrossVA} (i.e. when the data $\vy$ comes from a function $f$ that is a realization of a GP with covariance function that belongs to the family of covariance functions in section \ref{sec:covfunctions}).

However, there is no guarantee that the MLE method would perform optimally as this method is poorly robust with respect to model misspecifications. Besides, training and assessing the quality of a predictor should not be done on the same data (\citet{Tibshirani2009} in chapter 7). In particular, the MLE method does not show how well the model will do when it is asked to make new predictions for data it has not already seen. The CV method represents an alternative to estimate the covariance hyperparameters $(\sigma^2, \vtheta)$ for prediction purposes  \citep{Zhang2010,Bachoc2013CrossVA} and has the advantage of being more efficient and robust when the covariance function is misspecified \citep{Bachoc2013CrossVA}.

\subsection{Training model with Cross-Validation method for point-wise prediction}
\label{par:Cross-Valid}

We consider the same learning set of $n$ observations $\mathbf{D}_{\rm{learn}} =(\mathbf{X}, \vy)= \{ (\vx^{(i)} , y^{(i)}), \ i \in \{1,\ldots, n\} \}$ and we assume that the value of $\sigma^2_{\epsilon} \in [0,+\infty)$ is known. The Leave-One-Out method (i.e. $n$-Cross-Validation) consists in predicting $y^{(i)}$ by building a GP model, denoted $\mathcal{GP}_{-i}$ and trained on $\mathbf{D}_{-i}= \{ (\vx^{(j)} , y^{(j)})\}_{ j \in \{1,\ldots, n\} \setminus \{i\}}$. The obtained prediction mean and variance are functions of parameters $(\sigma^2, \vtheta)$ as shown in (\ref{eq:GPMean}) and (\ref{eq:GPVariance}) and are used to assess the predictive capability of the global GP model.

In the case of the Leave-One-Out method, the Mean Squared prediction Error (MSE) is used to assess the quality of the point-wise prediction (See \cite{WALLACH1989299} for more details about this metric) of the GP model, it can be expressed as 
\begin{equation}
\label{eq:LOO_MSE}
	 \mathcal{L O O}_{MSE}(\sigma^2, \vtheta) := \frac{1}{n} \sum_{i=1}^{n}\left(y^{(i)}-\tilde y_{i} \right)^{2},
\end{equation}
where $\tilde{y}_{i}$ and $\tilde \sigma_i^2$ are the Leave-One-Out predictive mean and variance of $f(\vx^{(i)})$ by a GP model trained on $\mathbf{D}_{-i}$ with the hyperparameters $(\sigma^2, \vtheta)$. 

{\it Hypothesis $\mathcal{H}_2$~: Let $({\bf e}_i)_{i=1}^n$ be the canonical basis of $\mathbb{R}^n$. We assume that ${\bf e}_i \not\in$ {\rm Im }$\mathbf{F}$ for all $i \in \{1,\ldots,n\}$. }

Let $\overline{\mathbf{K}}$ be the matrix defined by
\begin{equation}
\label{eq:defKbar}
    \overline{\mathbf{K}} = \mathbf{K}^{-1} - \mathbf{K}^{-1} \mathbf{F} \left( \mathbf{F}^{\top} \mathbf{K}^{-1} \mathbf{F} \right)^{-1} \mathbf{F}^{\top} \mathbf{K}^{-1} .
\end{equation}
For all $i \in \{1,\ldots,n\}$, we have $\left( \overline{\mathbf{K}} \right)_{i,i}>0$ by {Lemma \ref{lemma3} (see \ref{appendix:A})}, and, in the case of Ordinary or Universal Kriging, the Virtual Cross-Validation formulas \citep{Dubrule1983} of the predictive mean $\tilde y_i$ and variance $\tilde \sigma_i^2$  are given by
\begin{equation}
\label{eq:LOO_mean}
    y^{(i)}-\tilde y_{i} = \frac{ \left( \overline{\mathbf{K}}\vy \right)_i}{ \left( \overline{\mathbf{K}} \right)_{i,i} },
\end{equation}
and
\begin{equation}
\label{eq:LOO_var}
    \tilde \sigma^2_i  = \frac{1}{ \left( \overline{\mathbf{K}} \right)_{i,i}}.
\end{equation}

With the presence of the nugget effect, the GP regressor does not interpolate the training data $\vy$ but approximates them as best as possible. The Leave-One-Out method looks for the best approximation by minimizing the $\mathcal{L O O}_{MSE}$ criterion. The criterion (\ref{eq:LOO_MSE}) can be written with explicit quadratic forms in $\vy$
\begin{equation}
\label{eq:MSEsolution}
    (\hat{\sigma}_{MSE}^{2}, \hat{\vtheta}_{MSE}) \in \operatorname{argmin}_{\sigma^2, \vtheta} \ \vy^{\top} \overline{\mathbf{K}} \operatorname{Diag}\left(\overline{\mathbf{K}} \right)^{-2} \overline{\mathbf{K}} \ \vy. 
\end{equation}

Note that in the absence of the nugget effect $\sigma^2_{\epsilon}=0$, $\overline{\mathbf{K}}$ is of the form $\sigma^{-2}  \overline{\mathbf{R}}_{\vtheta}$ where $\overline{\mathbf{R}}_{\vtheta}$ does not depend on $\sigma^2$. The predictive variance $\hat{\sigma}_{MSE}^{2}$ can then be computed {by the following explicit quadratic form \citep{Bachoc2013CrossVA}}
\begin{equation}
\label{eq:sigmaCVMSE}
    \hat{\sigma}_{MSE}^{2} = {\frac{1}{n}} \ \vy^{\top} \overline{\mathbf{R}}_{\hat{\vtheta}_{MSE}} \operatorname{Diag}\left(\overline{\mathbf{R}}_{\hat{\vtheta}_{ MSE}}\right)^{-1} \overline{\mathbf{R}}_{\hat{\vtheta}_{MSE}} \boldsymbol{y},
\end{equation}
{and the optimal length-scale vector $\hat{\vtheta}_{MSE}$ is obtain by solving}
\begin{equation}
\label{eq:thetaMSE}
    \hat{\vtheta}_{MSE}  \in \operatorname{argmin}_{\vtheta} \  \vy^{\top} \overline{\mathbf{R}}_{\vtheta}  \operatorname{Diag}\left(\overline{\mathbf{R}}_{\vtheta}\right)^{-2} \overline{\mathbf{R}}_{\vtheta} \vy.
\end{equation}

{\subsection{Full-Bayesian approach}}
\label{subsec:Full_Bayes}

{In this subsection, we consider the full-Bayesian treatment of GP models \citep{WilliamsBayesianGP98}. Indeed,  the full-Bayesian approach  integrates the uncertainty about the unknown hyperparameters and assumes a \textit{prior} on the hyperparameters $(\sigma^2, \vtheta) \sim p(\sigma^2, \vtheta)$. Consequently, the probability density function (pdf) of the \textit{posterior} predictive distribution of $Y(\vx_{\rm{new}})$ at a new point $\vx_{\rm{new}}$ can be expressed as an integral over the hyperparameters (we omit the conditioning over inputs $\mathbf{X}$ and $\vx_{\rm{new}}$):}
\begin{equation}
\label{eq:defBayes}
    {p (y_{\rm{new}} \mid \vy )  = \iint p (y_{\rm{new}} \mid \vy, \sigma^2, \vtheta )  p( \sigma^2, \vtheta \mid \vy ) \, \mathrm{d}\sigma^2 \mathrm{d} \vtheta},
\end{equation}
{where $p ( y_{\rm{new}} \mid \sigma^2, \vtheta \big)$ is the pdf of $Y(\vx_{\rm{new}})$ given $\vy, \sigma^2$ and $\vtheta$, and $p ( \sigma^2, \vtheta \mid \vy ) \propto p( \vy \mid \sigma^2, \vtheta) p( \sigma^2, \vtheta ) $ is the hyperparameters' posterior distribution.}

{The implementation of the full-Bayesian approach requires the evaluation of the previous integral and the posterior $p ( \sigma^2, \vtheta \mid \vy )$, which cannot be computed directly. It is common to use Markov chain Monte Carlo (MCMC) methods for sampling and inference from the posterior distribution of the hyperparameters to overcome this issue, using, in particular, the Metropolis-Hastings (MH) algorithm  \citep{Robert_2004} or Hamiltonian Monte Carlo (HMC) \citep{Neal93,Neal_1996}.}

{Therefore, the predictive distribution is obtained by Monte Carlo}
\begin{equation}
    { p (y_{\rm{new}} \mid \vy )  \simeq \frac{1}{N} \sum_{i=1}^N p (y_{\rm{new}} \mid \vy, \sigma_i^2, \vtheta_i ), }
\end{equation}
{where $N$ denotes the MCMC sample size and $(\sigma_i^2, \vtheta_i)$ is the $i$-th sample drawn from the posterior distribution $p( \sigma^2, \vtheta \mid \vy)$.}

{Finally, one can draw a sample  $\big(Y_i(\vx_{\rm{new}})\big)_{i=1}^N$ of $Y(\vx_{\rm{new}})$ following the posterior distribution  $p(y_{\rm{new}}|\sigma_i^2,\vtheta_i)$ as in (\ref{eq:posteriorlaw}) for each $i=1,\ldots,N$ and build the Prediction Intervals $\mathcal{PI}_{1-\alpha}$ by taking the empirical quantiles of order $\alpha/2$ and $1-\alpha/2$ of the sample $\big(Y_i(\vx_{\rm{new}})\big)_{i=1}^N$.}

{Note that the \textit{plug-in} approaches (e.g. the MLE method in \ref{par:MLE}) consider  (\ref{eq:defBayes}) and replace $p( \sigma^2, \vtheta \mid \vy )$ by a Dirac distribution centered on a value such as $(\hat{\sigma}^2_{ML}, \hat{\vtheta}_{ML})$ that maximizes the likelihood function.} \\

\section{Prediction Intervals estimation for Gaussian Processes}
\label{sec4:PIforGP}

{Using the Cross-Validation method,} the MSE hyperparameters $(\hat{\sigma}_{MSE}^{2}, \hat{\vtheta}_{MSE})$ are obtained from a point-wise prediction metric and do not focus on Prediction Intervals neither on quantifying the uncertainty of the model. For these purposes, using the CP is more appropriate.

The \textit{Coverage Probability} (CP) is defined as the probability that the Prediction Interval procedure will produce an interval that captures what it is intended to capture \citep{Hong2009}. In the Leave-One-Out framework, we keep the notations of $\tilde y_{i}$ and $\tilde \sigma_i^2$ : the predictive mean and variance on $\vx^{(i)} \in \mathbf{X}$  using the learning set $\mathbf{D}_{-i}= \{ (\vx^{(j)} , y^{(j)})\}_{ j \in \{1,\ldots, n\} \setminus \{i\}}$. We define then the Leave-One-Out CP $\mathbb{\tilde P}_{1-\alpha}$ as the percentage of observed values $\vy$ belonging to Prediction Intervals $\mathcal{PI}_{1-\alpha}$ of $\tilde y_{i}$ for all $i \in \{1,\ldots,n\}$
\begin{equation}
    \begin{aligned}
        \mathbb{\tilde P}_{1-\alpha} &{= \frac{1}{n} \sum_{i=1}^{n} {\mathbf{1}} \{ y^{(i)} \in  \mathcal{PI}_{1-\alpha} ( \vx^{(i)}  )\} }, \\
        &= \frac{1}{n} \sum_{i=1}^{n} {\mathbf{1}} \{ \tilde y_{i} + q_{\alpha / 2} \times \tilde \sigma_{i} < \ y^{(i)} \leq \ \tilde y_{i} + q_{1-\alpha / 2} \times \tilde \sigma_{i} \} ,
    \end{aligned}
\end{equation}

where $q_a$ is the $a$-quantile of the standard normal {distribution} and ${\mathbf{1}}\{A\}$ is the indicator function of  $A$. We introduce the Heaviside step function $h$
\begin{equation}
    h(x) = {\mathbf{1}}\{x \geq 0\} = \left\{ \begin{array}{l l}
        1 & \quad \text{if $x \geq 0$}\\
        0 & \quad \text{if $x < 0$}\\ \end{array} \right.
\end{equation}

The Leave-One-Out CP $\mathbb{\tilde P}_{1-\alpha}$ can be written as
\begin{equation}
\label{eq:LOO_Palpha}
    \mathbb{\tilde P}_{1-\alpha} = \frac{1}{n} \sum_{i=1}^{n} h \left( q_{1-\alpha / 2} - \frac{y^{(i)}-\tilde y_{i}}{\tilde \sigma_{i} } \right)  - \frac{1}{n} \sum_{i=1}^{n} h \left( q_{\alpha / 2} - \frac{y^{(i)}-\tilde y_{i}}{\tilde \sigma_{i} } \right) .
\end{equation}

When the model is well-specified, the coverage of the Prediction Intervals $\mathcal{PI}_{1-\alpha}$ is optimal as the predictive distribution is fully characterized by the Gaussian {distribution} (see section \ref{sec:covfunctions}), each term of the right-hand side of (\ref{eq:LOO_Palpha}) is an unbiased estimator of the probability 
\begin{equation}
    \mathbb{P} \left( \frac{ Y(\vx^{(i)}) - \tilde y_i }{\tilde \sigma_i}  \leq  q_{1-\alpha / 2} \ \Big| \ \mathbf{D}_{-i} \right) = 1-\alpha / 2 ,
\end{equation}

and
\begin{equation}
    \mathbb{P} \left( \frac{ Y(\vx^{(i)}) - \tilde y_i }{\tilde \sigma_i}  \leq  q_{\alpha / 2} \ \Big| \ \mathbf{D}_{-i} \right) = \alpha / 2.
\end{equation}
Conversely, if the model is misspecified, each predictive quantile, needs to be quantified properly with respect to the normal {distribution} quantile so that the CP as described in section \ref{sec2:ProbForm} achieves the desired level.

Let $a \in (0,1/2) \cup (1/2,1)$ describe a nominal level of quantile. We define the \textit{quasi-Gaussian} proportion $\psi_a$ as a map from $[0,+\infty)\times(0,+\infty)^{d}$ to $[0,1]$ 
\begin{equation}
\label{eq:quasiGauss}
    \psi_a \left( \sigma^2, \vtheta \right) = \frac{1}{n} \sum_{i=1}^{n} h\left( q_a - \frac{ y^{(i)}-\tilde y_{i}}{\tilde \sigma_{i}} \right) ,
\end{equation}
where $\tilde{y_{i}}$ and $\tilde \sigma_{i}$ are the predictive mean and variance at $\vx^{(i)}$ using the learning set $\mathbf{D}_{-i}$ and the hyperparameters $(\sigma^2, \vtheta)$. 
Given the Virtual Cross-Validation formulas \citep{Dubrule1983}, $\psi_a$ can be written in terms of the covariance matrix $\overline{\mathbf{K}}$
\begin{equation}
\label{def:psi_a}
    \psi_a(\sigma^2, \vtheta) = \frac{1}{n} \sum_{i=1}^{n} h \left( q_a - \frac{ \left( \overline{\mathbf{K}}\vy \right)_i }{ \sqrt{ \left( \overline{\mathbf{K}} \right)_{i,i}}} \right) .
\end{equation}

The \textit{quasi-Gaussian} proportion $\psi_a$ describes how close the $a$-quantile $q_a$ of the standardized predictive distribution is to the level $a$ (ideally, it should correspond to $a$). Therefore, the objective is to fit the hyperparameters $(\sigma^2, \vtheta)$ according to the \textit{quasi-Gaussian} proportions and to find two pairs $(\overline{\sigma}^2, \overline{\vtheta})$ and $(\underline{\sigma}^2, \underline{ \vtheta})$ such that $\psi_{1-\alpha/2}(\overline{\sigma}^2, \overline{\vtheta}) = 1-\alpha/2$ and $\psi_{\alpha/2}(\underline{\sigma}^2, \underline{ \vtheta}) = \alpha/2$. This allows us to get the optimal Leave-One-Out CP by respecting the nominal confidence level $(1-\alpha)$, that is $\mathbb{\tilde P}_{1-\alpha} = 1 - \alpha$.
 
\subsection{Presence of nugget effect}
 
In this section, we assume $\sigma^2_{\epsilon}>0$. The \textit{quasi-Gaussian} proportion $ \psi_a$ is, however, piecewise constant and can take values only in the finite set $\{k/n, k \in\{0,\ldots,n\} \}$. We first need to modify the problem $\psi_a \left( \sigma^2, \vtheta \right) = a$. Let $\delta > 0$, we define the continuous functions $h^-_{\delta}$ and $h^+_{\delta}$
\begin{equation}
    \begin{aligned} h^+_{\delta}(x) &= \left\{ \begin{array}{l l}
        1 & \quad \text{if $x > \delta$}\\
        x/{\delta} & \quad \text{if $ 0 < x \leq \delta$}\\
        0 & \quad \text{otherwise}\\ \end{array} \right. \\
        h^-_{\delta}(x) &= \left\{ \begin{array}{l l}
        1 & \quad \text{if $x \geq 0$}\\
        1+x/{\delta} & \quad \text{if $ -\delta \leq x < 0$}\\
        0 & \quad \text{otherwise}\\ \end{array} \right.
    \end{aligned}
\end{equation}

If $a>1/2$ we define 
\begin{equation}
    \psi^{\left(\delta\right)}_a \left( \sigma^2, \vtheta \right) =\frac{1}{n} \sum_{i=1}^{n}  h^+_{\delta} \left( q_a - \frac{ \left( \overline{\mathbf{K}}\vy \right)_i }{ \sqrt{ \left( \overline{\mathbf{K}} \right)_{i,i}}} \right) .
\end{equation}

If $a<1/2$ we define 
\begin{equation}
    \psi^{\left(\delta\right)}_a \left( \sigma^2, \vtheta \right) =\frac{1}{n} \sum_{i=1}^{n}  h^-_{\delta} \left( q_a - \frac{ \left( \overline{\mathbf{K}}\vy \right)_i }{ \sqrt{ \left( \overline{\mathbf{K}} \right)_{i,i}}} \right) .
\end{equation}

Let $\delta >0$ be a small enough so that $\delta<q_a$ if $a>1/2$ (respectively, $\delta<q_{1-a}$ if $a<1/2$) in such a way that $h^+_{\delta}(q_a)=1$ (respectively, $h^-_{\delta}(q_a)=0$). We consider the problem
\begin{equation}
\label{Eq:OptimLOO}
    \psi^{(\delta)}_a \left( \sigma^2, \vtheta \right) = a ,
\end{equation}

and we denote by $\mathcal{A}_{a, \delta}$ the solution set of the problem (\ref{Eq:OptimLOO})
\begin{equation}
    \mathcal{A}_{a, \delta} := \left\{ (\sigma^2, \vtheta) \in [0,+\infty) \times (0,+\infty)^d , \ \psi^{(\delta)}_{a}(\sigma^2, \vtheta) = a \right\} . 
\end{equation}

{\it Hypothesis $\mathcal{H}_3$~:
{Let $k_{\epsilon} = \operatorname{Card} \{ i \in \{1,\ldots,n\}, \ \frac{\left(\boldsymbol{\Pi} \vy \right)_i}{\sqrt{\left(\boldsymbol{\Pi} \right)_{i i}}}  \leq \sigma_{\epsilon}q_a \}$ where $\boldsymbol{\Pi}$ is the orthogonal projection matrix on {\rm (Im }$\mathbf{F})^{\perp}$ such that $\boldsymbol{\Pi} = \mathbf{I}_n - \mathbf{F} \left( \mathbf{F}^{\top} \mathbf{F} \right)^{-1} \mathbf{F}^{\top}$. We assume that $k_{\epsilon}< na$ if $a > 1/2$ and $ k_{\epsilon} > na $ if $a < 1/2$}}.

\begin{remark}
    {The hypothesis $\mathcal{H}_3$ is typically satisfied in Ordinary and Universal Kriging. Indeed, $\boldsymbol{\Pi}$ is the projection on the space {\rm (Im }$\mathbf{F})^{\perp}$ and is expected to remove the trend of the model. It is reasonable to think that $\left(\boldsymbol{\Pi} \vy \right)$ is centered and that
    \begin{equation}
        \operatorname{Card} \{ i \in \{1,\ldots,n\}, \ \left(\boldsymbol{\Pi} \vy \right)_i \leq 0 \} \approx \frac{n}{2} .
    \end{equation}
    If $\sigma^2_{\epsilon}$ is smaller than $\sigma^2$, then we should also have
    \begin{equation}
        \operatorname{Card} \{ i \in \{1,\ldots,n\}, \ \frac{\left(\boldsymbol{\Pi} \vy \right)_i}{\sqrt{\left(\boldsymbol{\Pi} \right)_{i i}}}  \leq \sigma_{\epsilon} q_a \} \approx \frac{n}{2} ,
    \end{equation}
    so that the hypothesis $\mathcal{H}_3$ should be fulfilled.} 
\end{remark} 

\begin{proposition}
\label{prop:prop1}
    Let us assume the hypotheses $\mathcal{H}_1$, $\mathcal{H}_2$ and $\mathcal{H}_3$, then $\mathcal{A}_{a, \delta}$ is non-empty.
\end{proposition}

\begin{proof}
    In \ref{appendix:A}.
\end{proof}

The challenge now is to identify and choose wisely the optimal solutions $(\sigma_{\rm{opt}}^2, \vtheta_{\rm{opt}}) \in \mathcal{A}_{a, \delta}$. Some authors \citep{KhosraviMPIW} suggest the mean Prediction Intervals width ({MPIW}) of Prediction Intervals $\mathcal{PI}_{1-\alpha}$ as an additional constraint to reduce the set of solutions. However, this constraint may not work when dealing with quantile estimation because the lower bound of the corresponding interval may be infinite. 

Instead, we will compare these solutions with MLE's solution $(\hat{\sigma}_{ML}^2, \hat{\vtheta}_{ML})$ (subsection \ref{par:MLE}) or MSE-CV solution $(\hat{\sigma}_{MSE}^2, \hat{\vtheta}_{MSE})$ (subsection \ref{par:Cross-Valid}) and we will take the closest pair $(\sigma_{\rm{opt}}^2, \vtheta_{\rm{opt}})$ by using an appropriate notion of similarity between multivariate Gaussian distributions. Ideally, we aim to solve the following problem
\begin{equation}
\label{prob:OptimDist}
    \operatorname{argmin}_{(\sigma^2, \vtheta) \in \mathcal{A}_{a, \delta}} d^2\left( (\sigma^2, \vtheta), (\sigma_0^2, \vtheta_0) \right),
\end{equation}
where $d$ is a continuous similarity measure of hyperparameters $(\sigma^2, \vtheta)$ operating on the mean vector $\vm$ and covariance matrix $\mathbf{K}$, and $(\sigma_0^2, \vtheta_0) =(\hat \sigma_{ML}^2, \hat \vtheta_{ML})$ or $ (\hat\sigma_{MSE}^2, \hat\vtheta_{MSE})$ as described in (\ref{eq:OptimML}) or (\ref{eq:MSEsolution}).

The resolution of the problem (\ref{prob:OptimDist}) may be too costly and heavy to solve when the dimension is high, say $d \geq 10$. An alternative is to apply \textit{the relaxation} method where we redefine this optimization problem of $\vtheta$ from $(0,+\infty)^d$ to $(0,+\infty)$ by shifting the length-scale vector $\vtheta_0$ by a parameter $\lambda \in (0,+\infty)$.

Let $\vtheta_0$ denote a solution of the problems (\ref{eq:OptimML}) or (\ref{eq:MSEsolution}) and for $\lambda \in (0,+\infty)$, let $H_{\delta}(\lambda)$ denote the subset 
\begin{equation}
\label{eq:Hdelta}
    H_{\delta}(\lambda)= \{ \sigma^2 \in [0,+\infty) , \ \psi^{(\delta)}_{a}(\sigma^2, \lambda \vtheta_0)=a\} .
\end{equation} 

{\it Hypothesis $\mathcal{H}_4$~: The set-valued mapping (the so-called correspondence function) $H_{\delta}:(0,+\infty) \to \mathcal{P}((0,+\infty))$, where $\mathcal{P}(S)$ denotes the power set of a set $S$, is lower semi-continuous, that is, for all $\lambda \in (0,+\infty)$, for each open set $\mathcal{U}$ with $H_{\delta}(\lambda) \cap \mathcal{U} \neq \emptyset$, there exists a neighborhood $\mathcal{O}(\lambda)$ such that if $\lambda^* \in \mathcal{O}(\lambda)$ then $H_{\delta}(\lambda^*) \cap \mathcal{U} \neq \emptyset$.} 

In the kriging framework, $\sigma^2$ should be as small as possible to reduce the uncertainty of the model, a natural choice of $\sigma^2_{\rm{opt}}$ is
\begin{align}
    &\forall \lambda \in (0,+\infty) \ : \sigma^2_{\rm{opt}}(\lambda) := \min \{ \sigma^2 \in [0,+\infty) , \ \psi^{(\delta)}_{a}(\sigma^2, \lambda \vtheta_0) = a\} .
\end{align}

\begin{proposition}
\label{prop:prop2}
    The function $\lambda \mapsto \sigma^2_{\rm{opt}}(\lambda)$ is well-defined under hypotheses $\mathcal{H}_1$ to $\mathcal{H}_3$, and continuous on $(0,+\infty)$ under the additional hypothesis $\mathcal{H}_4$.
\end{proposition}

\begin{proof}
    In \ref{appendix:A}.
\end{proof}

Concerning the choice of $d$, one known similarity measure between probability distributions is the Wasserstein distance, widely used in optimal transportation problems (see Chapter 6 of \cite{Villani2009} for more details). In case of two Gaussian random distributions $\mathcal{N}(\vm_1, \mathbf{K}_1)$ and $\mathcal{N}(\vm_2,\mathbf{K}_2)$, the second Wasserstein distance is equal to
\begin{equation}
    W^2_2(\mathcal{N}(\vm_1,\mathbf{K}_1),\mathcal{N}(\vm_2,\mathbf{K}_2)) = \| \vm_1 - \vm_2 \|^2 +  \Tr \left( \mathbf{K}_1 + \mathbf{K}_2 - 2\sqrt{\mathbf{K}_1^{1/2} \mathbf{K}_2 \mathbf{K}_1^{1/2}} \right) ,
\end{equation}
where, in our setting, $\vm_1= \mathbf{F} \widehat{\boldsymbol{\beta}}_1 = \left(\mathbf{F}^{\top} \mathbf{K}_1^{-1} \mathbf{F}\right)^{-1} \mathbf{F}^{\top} \mathbf{K}_1^{-1} \vy$ and $\vm_2 = \mathbf{F} \widehat{\boldsymbol{\beta}}_2 = \left(\mathbf{F}^{\top} \mathbf{K}_2^{-1} \mathbf{F}\right)^{-1} \mathbf{F}^{\top} \mathbf{K}_2^{-1} \vy$.

Therefore, each pair $(\sigma^2, \vtheta)$ is associated to a Gaussian distribution $\mathcal{N}(\vm,\mathbf{K})$ and we define the similarity measure $d$ as
\begin{equation}
    d^2\left( (\sigma^2, \vtheta), (\sigma_0^2, \vtheta_0) \right) = W^2_2(\mathcal{N}(\vm, \mathbf{K}),\mathcal{N}(\vm_0,\mathbf{K}_0)).
\end{equation}

The choice of the second Wasserstein distance $d^2$ and $\sigma^2_{\rm{opt}}$ makes the Prediction Intervals $\mathcal{PI}_{1-\alpha}$ shorter without the need for an additional metric like the MPIW and {without modifying the distribution of the obtained model significantly. We will see in Section \ref{Section:5.2} that, empirically, the barycenters of Prediction Intervals are not far from the predictive means obtained by MLE or MSE-CV methods}.

The \textit{relaxed} optimisation problem in (\ref{Eq:OptimLOO}) for the quantile estimation is given by the problem $\mathcal{P}_{\lambda }$ 
\begin{equation}
\label{prob:relax_dWass}
     \mathcal{P}_{\lambda} : \quad \operatorname{argmin}_{\lambda \in (0,+\infty)} \ \mathcal{L}(\lambda) := d^2\left( (\sigma_{\rm{opt}}^2(\lambda), \lambda \vtheta_0 ), (\sigma_0^2, \vtheta_0) \right).
\end{equation}

\begin{proposition}
\label{prop:prop3}
    Under hypotheses $\mathcal{H}_1$ to $\mathcal{H}_4$, the function $\mathcal{L} : (0,+\infty) \ \rightarrow \sR^+$ is continuous and coercive on $(0, +\infty)$. The problem $\mathcal{P}_{\lambda}$ admits at least one global minimizer $\lambda^*$ in $(0,+\infty)$.
\end{proposition}

\begin{proof}
    See \ref{appendix:A}.
\end{proof}

\begin{figure}[!ht]
	\centering
	\includegraphics[width=0.95\textwidth]{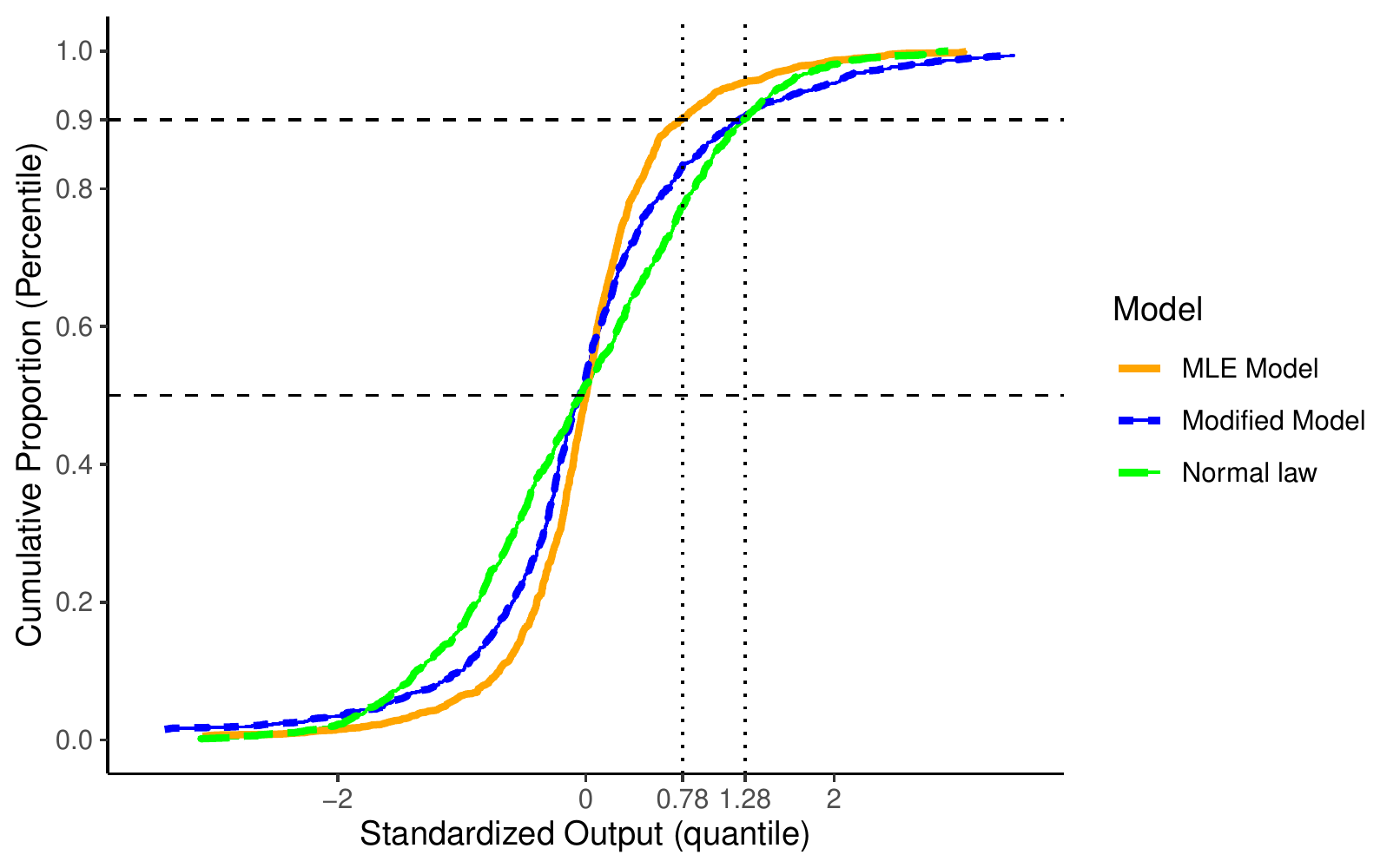}
	\caption{Illustration of the relaxation effect on the ECDF of the Leave-One-Out standardized predictive distribution on the quantile of level $a=90\%$; The relaxed standardized predictive distribution coincides with the standard normal {distribution} distribution on point $(q_a, a)=(1.28, 0.90)$ instead of $(\psi_a, a)=(0.78, 0.90)$ \\ 
	Green : standard normal {distribution} - MLE standardized Predictive distribution when the model is well-specified; Orange : MLE standardized predictive distribution when the model is misspecified; Blue : standardized predictive distribution after relaxing model's hyperparameters.} 
	\label{fig:relaxeffect}
\end{figure}

\begin{remark}
    The coercivity of the function $\mathcal{L}$ is guaranteed by the hypotheses $\mathcal{H}_1$ to $\mathcal{H}_3$ (see \ref{appendix:A}). The function $\mathcal{L}$ is also upper semi-continuous \citep{ZHAO1997}. The hypothesis $\mathcal{H}_4$ insures that $\mathcal{L}$ is continuous and that there exists a global minimizer. This hypothesis is not easy to check. If it does not hold or if it cannot be checked, then it is possible to solve the problem (\ref{prob:relax_dWass}) on a regular grid by a grid search method.
\end{remark}

Let $\widehat{\boldsymbol{\beta}}_{\rm{opt}}$ denote the corresponding regression parameter
\begin{equation}
    \widehat{\boldsymbol{\beta}}_{\rm{opt}}= \left(\mathbf{F}^{\top} \mathbf{K}_{\sigma_{\rm{opt}}^2(\lambda^*), \lambda^* \vtheta_0}^{-1} \mathbf{F}\right)^{-1} \mathbf{F}^{\top} \mathbf{K}_{\sigma_{\rm{opt}}^2(\lambda^*), \lambda^* \vtheta_0)}^{-1} \vy.
\end{equation}

The purpose of this resolution is to create a GP model with hyperparameters $(\widehat{\boldsymbol{\beta}}_{\rm{opt}}, \sigma_{\rm{opt}}^2(\lambda^*), \lambda^* \vtheta_0)$ able to predict the quantile $\tilde{y}_a$ such that a proportion $a$ of true values are below $\tilde{y}_a$ with respect to the constraint of \textit{quasi-Gaussian} proportion $\psi_a$ (see Figure \ref{fig:relaxeffect}). Finally, the Prediction Intervals $\mathcal{PI}_{1-\alpha}$ will be obtained using two GP models built with the same method, one for the upper quantile $1-\alpha/2$ with optimal relaxation parameter $\overline{\lambda}^*$ and the other for the lower quantile of $\alpha/2$ with parameter $\underline{\lambda}^*$. The CP of $\mathcal{PI}_{1-\alpha}$ is optimal and insured by respecting the coverage of each quantile as shown in (\ref{eq:LOO_Palpha}). In the following, we call this method \textit{Robust Prediction Intervals Estimation} (RPIE).

\subsection{Absence of nugget effect}

When the nugget effect is null $\sigma^2_{\epsilon}=0$, the set of solutions $\mathcal{A}_{a, \delta}$ is still non-empty because one can show that, for $\vtheta$ in the neighborhood of $\boldsymbol{0} \in \sR^d$, the problem $\psi^{(\delta)}_{a}(\sigma^2, \vtheta) = a$ has a solution $\sigma^2 \in (0,+\infty)$ (see \ref{appendix:B}). In particular, the correspondence function $H_{\delta}$ is non-empty valued for $\lambda>0$ small enough and it may be empty-valued for some large $\lambda \in (0,+\infty)$. 
{We may think, however, that $H_{\delta}$ is non-empty valued and that $\sigma^2_{\rm{opt}}(\lambda)$ exists for $\lambda$ close to one. 
Indeed, assume for a while that the model is well-specified, that is, there exist hyperparameters $(\boldsymbol{\beta}_*, \sigma^2_*, \vtheta_*)$ such that $\vy$ corresponds to  a realization of a random vector $\textbf{\textit{Y}} \sim \mathcal{N}( \mathbf{F} \boldsymbol{\beta}_*, \ \sigma^2_* \mathbf{R}_{\vtheta_*})$.
The existence of $H_{\delta}(\lambda)$ and $\sigma^2_{\rm{opt}}(\lambda)$ depend on the condition $k_{\lambda} \leq na$, where $k_{\lambda}$ is the integer defined by
\begin{equation}
    k_{\lambda} := \operatorname{Card} \left\{i \in \{1,\ldots,n\} ,  \left( \overline{\mathbf{R}}_{\lambda \vtheta_0} \vy \right)_i \leq 0 \right\} .
\end{equation}
Since $\overline{\mathbf{R}}_{\vtheta_*} \textbf{\textit{Y}}$ is centered,
we can anticipate that
\begin{equation}
    \operatorname{Card} \left\{i \in \{1,\ldots,n\} ,  \left( \overline{\mathbf{R}}_{\vtheta_*} \vy \right)_i \leq 0 \right\} \approx \frac{n}{2} .
\end{equation} 
Hence, the condition $n/2 < k_{\lambda} \leq na$ should be satisfied in a neighborhood of $\lambda=1$ since $\vtheta_0$ should be close to $\vtheta_*$. Finally, eventhough the function $\mathcal{L}$ is not defined on $(0,+\infty)$, we can solve (\ref{prob:relax_dWass}) by a grid search method on its domain.}

\section{Numerical Results}
\label{sec5:Results}

\subsection{{Test cases with analytical functions}}

{In this section, we give three numerical examples to illustrate Prediction Intervals estimation by the RPIE method. We show that for the \textit{Wing-Weight} function, the model is well-specified as  the CP is optimal for different levels, hence, no robust calibration of Prediction Intervals is required. However, for \citet{zhou1998adaptive} and \citet{Morokoff95} functions where the model is misspecified and for a given confidence level $\alpha$, we apply the RPIE method as described in section \ref{sec4:PIforGP} to estimate both upper and lower bounds of Predictions Intervals. The following metrics : the Leave-One-Out CP $\tilde{\mathbb{P}}_{1-\alpha}$ defined in (\ref{eq:LOO_Palpha}), the Coverage Probability (CP), the mean ($\text{MPIW}$) and standard-deviation ($\text{SdPIW}$) of the Prediction Interval width, and the accuracy $Q^2$ \citep{Kleijnen2000AMF} are used to assess and compare GP models built by MLE or MSE-CV methods, full Bayesian approach or the RPIE method. They can be used either for point-wise prediction comparison ($Q^{2}$ will be given in some cases for information, it does not represent the main metric of this section):}

\begin{equation}
	Q^{2}=1-\frac{\sum_{i=1}^{n_{test}}\left( y_{test}^{(i)}-\tilde{y}_{i,test}\right)^{2}}	{\sum_{i=1}^{n_{test}}\left( y_{test}^{(i)}-\overline{y}\right)^{2}} ,
	\label{eq:Q2}
\end{equation}
{or for quantifying the \textit{goodness} of Prediction Intervals:}
\begin{equation}
	\tilde{\mathbb{P}}_{1-\alpha} = \frac{1}{n} \sum_{i=1}^{n} {\mathbf{1}} \{ y^{(i)} \in  \mathcal{PI}_{1-\alpha} ( \vx^{(i)}  )\},
\end{equation}

\begin{equation}
	\text{CP}_{1-\alpha} = \frac{1}{n_{test}} \sum_{i=1}^{n_{test}} {\mathbf {1}} { y_{test}^{(i)} \in  \mathcal{PI}_{1-\alpha} \left( \vx^{(i)}_{test}  \right)} ,
\end{equation}

\begin{equation}
	\text{MPIW}_{1-\alpha} = \frac{1}{n_{test}} \sum_{i=1}^{n_{test}} \left| \mathcal{PI}_{1-\alpha} \big( \vx^{(i)}_{test} \big) \right| ,
\end{equation}

and,
\begin{equation}
	\text{SdPIW}_{1-\alpha} = \sqrt{\frac{1}{n_{test}} \sum_{i=1}^{n_{test}} \left[ \big| \mathcal{PI}_{1-\alpha} \big( \vx^{(i)}_{test}  \big) \big| - \text{MPIW}_{1-\alpha} \right]^2 } ,
\end{equation}
{where $\vy_{test}=\left(y_{test}^{(1)},\ldots,y_{test}^{(n_{test})} \right)$ is the vector to predict at $\left(\vx^{(1)}_{test},\ldots,\vx^{(n_{test})}_{test} \right)$, $\mathcal{PI}_{1-\alpha}$ is the $(1-\alpha) \times 100 \%$ confidence Prediction Interval delimited by the quantiles $q_{1-\alpha/2}$ and $q_{\alpha/2}$, and $\left| \mathcal{PI}_{1-\alpha} \right|$ is the length of the interval.}

{Note that the ${\rm CP}_{1-\alpha}$ may be different from the Leave-One-Out CP $\tilde{\mathbb{P}}_{1-\alpha}$, this case can happen when the distributions of the training and testing sets are different. However, a Leave-One-Out CP $\tilde{\mathbb{P}}_{1-\alpha}$ close to ${1-\alpha}$ insures that, if the assumption of \textit{i.i.d} distributions is respected i.e. $p_{train}(\mathbf{X}, \boldsymbol{y}) = p_{test}(\mathbf{X}, \boldsymbol{y})$, $\text{CP}_{1-\alpha}$ should be also close to ${1-\alpha}$}.

{This subsection provides results obtained on $d=10$-dimensional GP with constant mean function (Ordinary Kriging). The value of $\delta$ is fixed at $\delta=10^{-2}$. We implement our methods using the package \textit{kergp} \citep{Roustant2020} on R. For the computational time, we use an Intel(R) Core(TM) i5-9400H CPU @ 2.50GHz with a RAM of 32 Go.}

\subsubsection*{{Example 1: Well-specified model - The Wing Weight function}}

{The Wing Weight function is a model in dimension $d=10$ proposed by \cite{Forrester2008} that estimates the weight of a light aircraft wing. For an input vector $\vx \in \sR^{10}$, the response $y$ is:
\begin{equation}
\label{eq:wingweight}
    f(\vx) = 0.036 x_{1}^{0.758} x_{2}^{0.0035}\left(\frac{x_{3}}{\cos ^{2} \left( x_{4} \right)}\right)^{0.6} x_{5}^{0.006} x_{6}^{0.04}\left(\frac{100 \ x_{7}}{\cos{(x_{4})}}\right)^{-0.3}\left(x_{8} x_{9}\right)^{0.49} + x_{1} x_{10}.
\end{equation}}

{The components $x_i$ are assumed to vary over the ranges given in Table \ref{table:inputRange} (see \cite{Forrester2008} and \cite{Moon2010} for details)}

\begin{table*}[htp!]
    \begin{center}
    	\centering 
    	\caption{The input variables $x_j$ and their domain ranges $[a_j; b_j]$}
    	\label{table:inputRange} 
    	\begin{tabular}{ c c c c c } 
			\toprule
			\addlinespace
    	    Component & Domain &  &  Component & Domain\\  
    	    \midrule 
        	\addlinespace
        	$x_1$ & $[150; 200]$&  & $x_6$ & $[0.5; 1]$\\ 
        	$x_2$ & $[220; 300]$&  & $x_7$ & $[0.08; 0.18]$\\ 
        	$x_3$ & $[6, 10]$&  & $x_8$ & $[2.5; 6]$\\ 
        	$x_4$ & $[-10; 10]$&  & $x_9$ & $[1700, 2500]$\\ 
        	$x_5$ & $[16; 45]$&  & $x_{10}$ & $[0.025; 0.08]$\\  
        	\addlinespace
        	\bottomrule
        \end{tabular}
    \end{center}
\end{table*}


{We create an experimental design $\mathbf{X}$ of $n=600$ observations and $d=10$ variables where observations $\vx^{(i)}=\left(x^{(i)}_1,\ldots,x^{(i)}_d\right)$ are sampled i.i.d with uniform distribution over $\bigotimes_{j=1}^d [a_j,b_j]$. We generate the response $\vy=\left(y^{(1)},\ldots,y^{(n)}\right)$ such that $y^{(i)}=f (\vx^{(i)}) + \epsilon^{(i)}$ with $f$ defined in (\ref{eq:wingweight}) and $\epsilon^{(i)}$ are sampled i.i.d. with the distribution $\mathcal{N}(0,\sigma^2_\epsilon=25)$. Here the nugget effect is estimated with the methodology described in \cite{Iooss2017} and the covariance kernel is the Mat\'ern 3/2.}

\begin{table*}[bp!]
	\begin{center}
		\centering
		\caption{Performances of methods (MLE, MSE-CV and Full-Bayesian) for Wing Weight function}
		\label{tab:WingWeight_Results} 
		\begin{tabular}{ l c c c c c}
			\toprule &  \multicolumn{2}{c}{Before RPIE} & \multicolumn{2}{c}{After RPIE} & Full-Bayesian \\ 
			\toprule 
			\addlinespace 
			& MLE  & MSE-CV & MLE  & MSE-CV & - \\ 
			\midrule 
			\addlinespace 
			$Q^2$   &  0.563  &  0.764 &  n.c  &  n.c  &  0.562  \\
			\addlinespace
			\addlinespace
			$\mathbb{\tilde P}_{99\%}$ & 99.1  & 99.8 & 98.9 & 98.9 &  99.1 \\
			$\text{CP}_{99\%}$ & 98.7 & 100  &  98.7  &  98.0  & 98.7 \\
			$\mathbb{\tilde P}_{95\%}$ & 94.0  & 98.9 &  94.9 & 94.9 &  94.2 \\
			$\text{CP}_{95\%}$ & 95.3 & 99.3  &  96.7 &  96.0 &  95.3 \\
			\addlinespace
			\addlinespace
			$\mathbb{\tilde P}_{90\%}$ & 90.1  & 96.9 &  90.0   & 90.0  & 90.9  \\
			$\text{CP}_{90\%}$ & 91.3 & 96.0  &  89.3 &  90.0  & 91.3  \\
			\addlinespace
			\addlinespace
			Ct & 2min 12s & 32min 42s  &  6min$^*$ & 37min$^*$  & 4h 39min  27s \\
			\bottomrule
		\end{tabular}
	\end{center}
	\begin{tablenotes}[flushleft]
		\scriptsize
		\item $Q^2$: Accuracy; $\mathbb{\tilde P}_{1-\alpha}$: The Leave-One-Out CP in \% on the training set; $\text{CP}_{1-\alpha}$: The CP in \% on the testing set and Ct: computational time.
		\item *: The approximated cumulative computational time after running the RPIE method for all levels.
	\end{tablenotes}
\end{table*}

{The GP model is trained on $75\%$ of the data ($25\%$ of data is left for testing). The diagnostics of the model are presented in Table \ref{tab:WingWeight_Results} with the metrics described above. The accuracy $Q^2$ is moderate for MLE and Full-Bayesian methods. The MSE-CV does much better, an expected result since the MSE-CV method is more adapted for point-wise prediction criterion. However, the Leave-One-Out CP $\mathbb{\tilde P}_{1-\alpha}$ for two different levels $\alpha=5\%, 10\%$ is far from the required level, which means that they were poorly estimated with point-wise prediction criterion. In addition, Table \ref{tab:WingWeight_Results} shows in particular that the model is well-specified for Mat\'ern 3/2 correlation kernel with the MLE method since the CPs are optimal and close to the required level. This claim is empirical and can be verified either by comparing the standardized predictive distribution with the standard normal distribution as in Figure \ref{fig:relaxeffect} or using \cite{Shapiro65} normality test (in this example, $p$-value $= 0.203$). The Full-Bayesian approach also does well in estimating Prediction Intervals in the case of a well-specified model. Indeed, the hyperparameters' posterior distribution $p(\sigma^2, \vtheta \mid \vy)$ is concentrated around the MLE estimator, so the plug-in MLE approach and the Full-Bayesian approach give similar predictive distributions and Prediction Intervals. However, its computational time is extremely long compared to other methods (e.g. 100 times longer than the MLE method). Concerning the RPIE method, one can notice that it provides the optimal coverage at each required level, either on training or testing sets. However, we do not see significant interest in applying it here (except for the MSE-CV solution).}

{Example 1 is a case of well-specified model in which the CPs obtained by the MLE method satisfy the nominal value and the RPIE method does not bring a significant additional value (at least for the MLE solution).}

\subsubsection*{{Example 2: Misspecified model with noise - Morokoff \& Caflisch function -} }

{We consider the \cite{Morokoff95} function defined on $[0,1]^d$ by
\begin{equation}
\label{eq:Morokoff}
    f(\vx) = \frac{1}{2} \Big(1+\frac{1}{d}\Big)^d \prod_{i=1}^d (x_i)^{1/d}.
\end{equation}}

{In Example 2, we consider an experimental design $\mathbf{X}$ of $n=600$ observations and $d=10$ correlated inputs. Each observation has the form $\vx^{(i)}=\Big( \Phi(z^{(i)}_1),\ldots,\Phi(z^{(i)}_d)\Big) \in \sR^d$, $\Phi$ is the CDF of the standard normal distribution, $\vz^{(i)}$ are sampled from the multivariate distribution $\mathcal{N}(\boldsymbol{0}, \mathbf{C})$ and $\mathbf{C} \in \sR^{d \times d}$ is the following covariance matrix:
\begin{equation*}
     \mathbf{C} = \begin{bmatrix}
        1 & 0.90 & 0 & 0 & 0 & 0.50 & -0.30 & 0 & 0 & 0 \\
        0.90 & 1 & 0 & 0 & 0 & 0 & 0 & 0.10 & 0 & 0 \\
        0 & 0 & 1 & 0 & -0.30 & 0.10 & 0.40 & 0 & 0.05 & 0 \\
        0 & 0 & 0 & 1 & 0.40 & 0 & 0 & -0.35 & 0 & 0 \\
        0 & 0 & -0.30 & 0.40 & 1 & 0 & 0 & 0 & 0.10 & 0 \\
        0.05 & 0 & 0.10 & 0 & 0 & 1 & 0 & 0 & 0 & 0 \\
        -0.30 & 0 & 0.40 & 0 & 0 & 0 & 1 & 0 & 0 & -0.30 \\
        0 & 0.1 & 0 & -0.35 & 0 & 0 & 0 & 1 & 0 & 0\\
        0 & 0 & 0.05 & 0 & 0.10 & 0 & 0 & 0 & 1 & 0\\
        0 & 0 & 0 & 0 & 0 & 0.& -0.3 & 0 & 0 & 1\\
            \end{bmatrix}.
\end{equation*}}

{The response vector $\vy$ is generated as $y^{(i)}= f(\vx^{(i)})+\epsilon^{(i)}$ with $f$ the \textit{Morokoff \& Caflisch} function defined in (\ref{eq:Morokoff}) and $\epsilon^{(i)}$ are sampled i.i.d. with the distribution $\mathcal{N}(0,\sigma^2_{\epsilon}=10^{-4})$. We consider the Mat\'ern anisotropic geometric correlation model with smoothness 5/2 as covariance model and we study the Prediction Interval's problem with a nugget effect estimated with the methodology \cite{Iooss2017}.}

\begin{table*}[h]
	\begin{center}
		\centering
		\caption{Performances of methods before and after RPIE for Morokoff \& Caflisch function; here  $1-\alpha=90\%$}
		\label{tab:Morokoff_Results}
		\begin{tabular}{ l c c c c c}
			\toprule &  \multicolumn{2}{c}{Before RPIE} & \multicolumn{2}{c}{After RPIE} & Full-Bayesian \\ 
			\toprule 
			\addlinespace 
			& MLE  & MSE-CV & MLE  & MSE-CV & - \\ 
			\midrule 
			\addlinespace 
			 $Q^2$ &  0.892 &  0.895  &  n.c  &  n.c &  0.891  \\ 
			\addlinespace 
			\addlinespace 
			$\mathbb{\tilde P}_{1-\alpha}$ & 93.6 &  98.3  &  90.0  &  90.0  & 93.8\\
			$\text{CP}_{1-\alpha}$ & 94.0 &  98.0  &  92.6  &  87.3  &  93.3 \\
			\addlinespace 
			\addlinespace
			$\text{MPIW}_{1-\alpha}$ & $1.68 \ 10^{-1}$ & $1.81 \ 10^{-1}$  &  $5.51 \ 10^{-2}$ &  $5.78 \ 10^{-2}$ & $1.66 \ 10^{-1}$ \\
			$\text{SdPIW}_{1-\alpha}$  & $9.61 \ 10^{-3}$ &  $4.16 \ 10^{-2}$ & $1.29 \ 10^{-2}$  & $1.41 \ 10^{-2}$  &  $9.27 \ 10^{-3}$ \\
			\addlinespace 
			\addlinespace
			Ct &  1min 16s  & 24min 18s  & 3min 55s  & 27min 43s &  4h 43min 38s \\ 
			\bottomrule
		\end{tabular}
	\end{center}
	\begin{tablenotes}[flushleft]
	\scriptsize
	\item $Q^2$: Accuracy; $\mathbb{\tilde P}_{1-\alpha}$: The Leave-One-Out CP in \%  on the training set; $\text{CP}_{1-\alpha}$: CP in \% on the testing set; $\text{MPIW}$: Mean of Prediction Interval widths; $\text{SdPIW}$: standard deviation of Prediction Interval widths and Ct: computational time.
	\end{tablenotes}
\end{table*}

{The model is not well-specified as Example 1 and the \citet{Shapiro65} test gives $p$-value $=1.253 \ 10^{-7}$. Table \ref{tab:Morokoff_Results} summarizes the results of MLE and MSE-CV estimations before and after applying the RPIE, compared with the Full-Bayesian approach. The accuracy $Q^2$ of both models is satisfactory and is {slightly} improved when using the MSE-CV method. However, before applying the RPIE, the Prediction Intervals are overestimated for both models. The CP does not correspond to the required level of $90\%$, and the MSE-CV model performs even worse. We note that the Full-Bayesian approach does not improve the quality of estimated Prediction Intervals for the same reason as explained before: the hyperparameters' posterior distribution $p(\sigma^2, \vtheta \mid \vy)$ is concentrated around the MLE estimator and the performances of both approaches are similar. We will see that this claim is also valid in Example 3.}

\begin{figure}[h]
    \centering
    \begin{subfigure}[b]{0.45\textwidth}
        \centering
        \includegraphics[width=\textwidth]{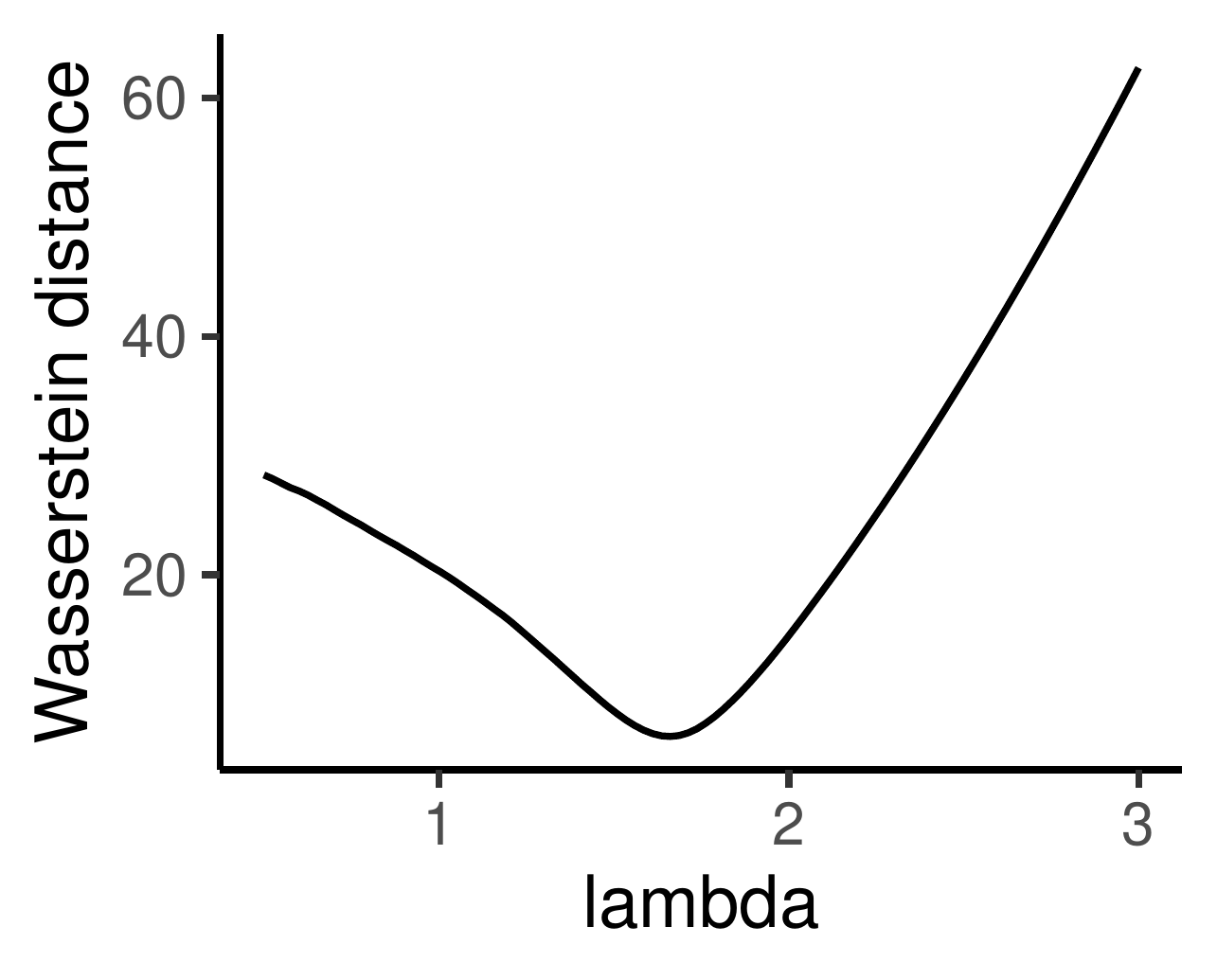}
        \caption{MLE solution $\vtheta_0=\hat{\vtheta}_{ML}$}
        \label{fig:ML_Solution}
    \end{subfigure}
    \hfil
    \begin{subfigure}[b]{0.45\textwidth}
        \centering
        \includegraphics[width=\textwidth]{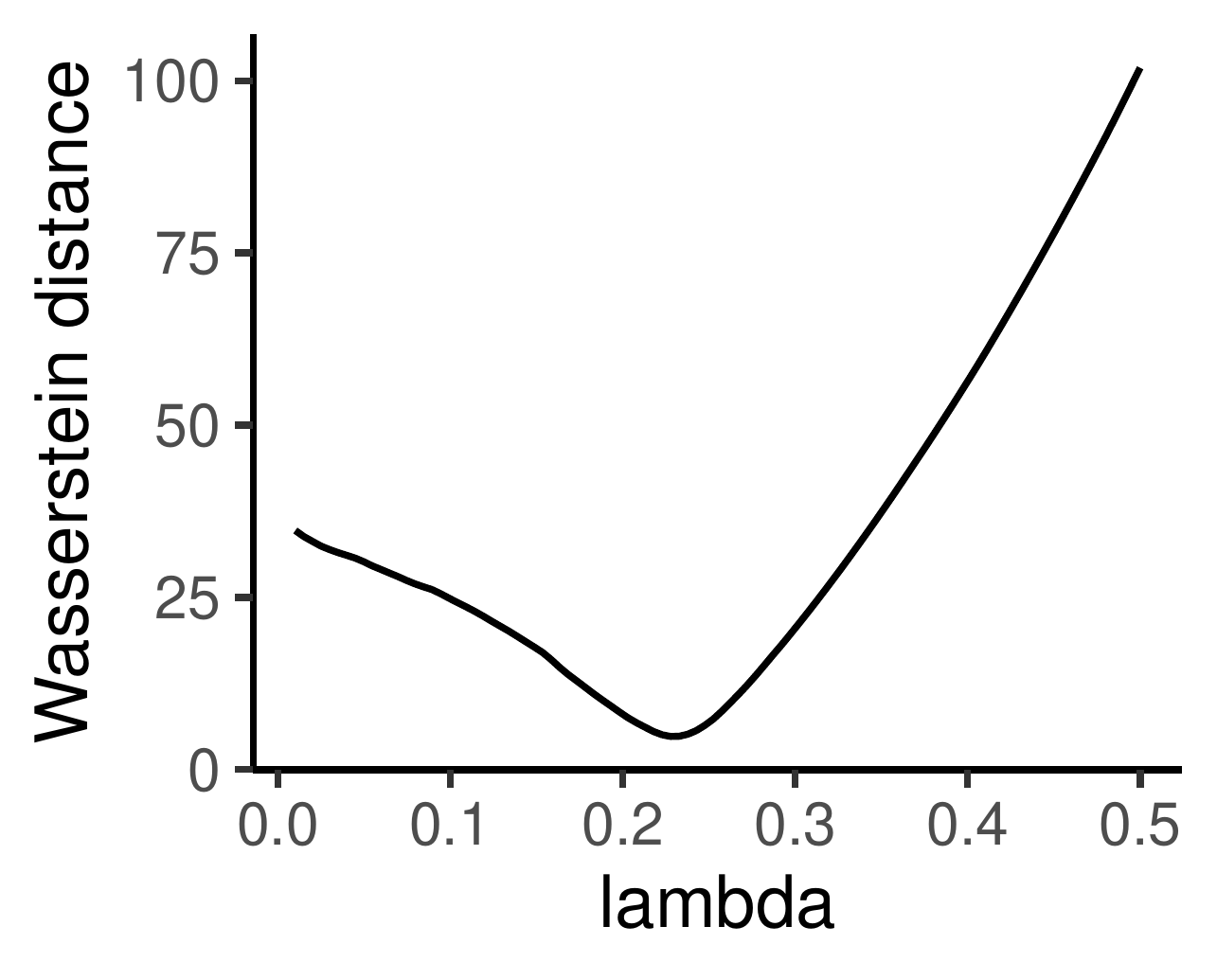}
        \caption{MSE-CV solution $\vtheta_0=\hat{\vtheta}_{MSE}$ }
        \label{fig:MSE_solution}
    \end{subfigure}
    \caption{The variation of the \textit{relaxed} Wasserstein distance $\mathcal{L}$ for \textit{Morokoff \& Caflisch} function; $ a= 1-\alpha/2 = 95\%$ }
    \label{fig:shifting_Morokoff}
\end{figure}

{We now address the problem of Prediction Intervals Estimation for each solution of MLE $\hat{\vtheta}_{ML}$ and MSE-CV $\hat{\vtheta}_{MSE}$. We consider the upper and lower bounds $1-\alpha/2=95\%$ and $\alpha/2=5\%$ and we apply the RPIE method as described in section \ref{sec4:PIforGP}. The optimal values $\overline{\lambda}^*$ and $\underline{\lambda}^*$ obtained from the resolution of the problem (\ref{prob:relax_dWass}) are used to build two GP models to estimate each bound. Figure \ref{fig:shifting_Morokoff} shows the variation of the function $\mathcal{L}$ for \textit{Morokoff \& Caflisch} example while solving the problem (\ref{prob:relax_dWass}) on the upper bound $1-\alpha/2=95\%$, it illustrates the statement of Proposition \ref{prop:prop3} : $\mathcal{L}$ is continuous and coercive on $(0,+\infty)$ and reaches a global minimum.}

{We consider now the Prediction Intervals built according to the RPIE method. In Table \ref{tab:Morokoff_Results}, one observes that these Prediction Intervals are three times shorter than those of MLE, MSE-CV models or Full-Bayesian approaches and have appropriate variances (e.g. more heterogeneous than MLE or Full-Bayesian method's Prediction Intervals). The coverage rate of $1-\alpha=90\%$ on the training set is achieved, which is the main objective of the RPIE method, and the CP on the testing set is very close to this level. Concerning the computational time, it appears that applying the RPIE method to MLE or MSE-CV solutions counts for a short computational time (only a few minutes to run in this example). The Full-Bayesian approach is still computationally heavy, as already discussed in the previous example and section \ref{sec1:Intro}.}

{Example 2 is a case of misspecified model with noise in which the CP obtained by MLE, MSE-CV {and Full-Bayesian} methods are not good. The RPIE method {fulfills its purpose}: its reduces Prediction Intervals width and improves the robustness of Prediction Intervals in such a way that they achieve the optimal coverage rate.}

\subsubsection*{{Example 3: Misspecified model without noise - Zhou function - } }

{The \cite{zhou1998adaptive} function, considered initially for the numerical integration of spiky functions, is defined on $[0,1]^d$ by
\begin{equation}
\label{eq:Zhou}
    f(\vx) = \frac{10^d}{2} \Big[ \phi \Big( 10\big(\vx-\frac{1}{3}\big)\Big) + \phi \Big( 10\big(\vx-\frac{2}{3}\big)\Big) \Big],
\end{equation}
where
\begin{equation}
    \phi(\vx) = (2\pi)^{-d/2} \exp \left(-0.5 \| \vx \|^2 \right).
\end{equation}}

{In Example 3, we create an experimental design $\mathbf{X}$ similar to Example 1, containing $n=600$ and $d=10$ variables where observations $\vx^{(i)}=\left(x^{(i)}_1,\ldots,x^{(i)}_d\right)$ are sampled independently with uniform distribution over $[0,1]^d$. As the \textit{Zhou} function in (\ref{eq:Zhou}) takes some high values, we generate the response $\vy$ by applying a logarithmic transformation: 
\begin{equation}
    y^{(i)} = \log{f(\vx^{(i)})}/ (d \log{10}).
\end{equation}}

{Note that there is no measurement noise here. We will address two situations: In the first setting, we assume that we know that there is no measurement noise, we impose that there is no nugget effect in the model $\sigma^2_{\epsilon}=0$ and we consider the Exponential anisotropic geometric correlation model ($\nu=1/2$) as covariance model. In the second setting, we assume that we do not know whether there is measurement noise and we estimate the nugget effect of the model. We consider consequently the Mat\'ern 3/2 anisotropic geometric correlation model ($\nu=3/2$), a reasonable choice for a smooth covariance model when assuming a nugget effect (See \ref{appendix:B} for further discussion).}

\begin{table*}[h!]
	\begin{center}
		\centering
		\caption{Performances of methods for \textit{Zhou} function (\ref{eq:Zhou}) in the first setting ($\sigma^2_{\epsilon}=0$) ; here $1-\alpha=90\%$}
		\label{tab:Zhou_Results}
        \begin{tabular}{ l c c c c c}
			\toprule &  \multicolumn{2}{c}{Before RPIE} & \multicolumn{2}{c}{After RPIE} & Full-Bayesian \\ 
			\toprule 
			\addlinespace 
			& MLE  & MSE-CV & MLE  & MSE-CV & - \\ 
			\midrule 
			\addlinespace 
			 $Q^2$ &   0.947 &  0.947  &  n.c  &  n.c &  0.948  \\ 
			\addlinespace 
			\addlinespace 
			$\mathbb{\tilde P}_{1-\alpha}$ & 92.0 &  42.1  &  90.0  &  90.0  & 92.0\\
			$\text{CP}_{1-\alpha}$ & 92.7 &  45.3  &  90.0  &  88.0  &  92.9 \\
			\addlinespace 
			\addlinespace
			$\text{MPIW}_{1-\alpha}$ & $4.60 \ 10^{-1}$ & $1.46 \ 10^{-1}$  &  $4.35 \ 10^{-1}$ &  $4.32 \ 10^{-1}$ & $4.59 \ 10^{-1}$ \\
			$\text{SdPIW}_{1-\alpha}$  & $1.06 \ 10^{-1}$ &  $3.48 \ 10^{-2}$ & $1.00 \ 10^{-1}$  & $1.00 \ 10^{-1}$  &  $1.08 \ 10^{-1}$ \\
			\addlinespace 
			\addlinespace
			Ct &  10s  & 31min 2s  & 2min 31s  & 33min 32s &  4h 56min 15s \\

			\bottomrule
		\end{tabular}
    \end{center}
    \begin{tablenotes}[flushleft]
	\scriptsize
	\item $Q^2$: Accuracy; $\mathbb{\tilde P}_{1-\alpha}$: The Leave-One-Out CP in \%  on the training set; $\text{CP}_{1-\alpha}$: The CP in \% on the testing set; $\text{MPIW}$: Mean of Prediction Interval widths; $\text{SdPIW}$: standard deviation of Prediction Interval widths and Ct: computational time.
    \end{tablenotes}
\end{table*}	

{In Table \ref{tab:Zhou_Results}, the models are good in terms of accuracy $Q^2$ with a small advantage for the {Full-Bayesian} approach, but none of them satisfies the required level of CP, especially the MSE-CV model with an extremely low CP. As we do not estimate the nugget effect in this setting, the computational time of the MLE method is low (a few seconds) where the RPIE still takes a couple of minutes, as in Example 2. We will notice (also in the industrial application) that the computational time after the RPIE method is generally twice to three times the computational time of MLE method when there is a nugget effect.}

{When proceeding similarly as Example 2 to build robust Prediction Intervals by the RPIE model, the result is striking in Table \ref{tab:Zhou_Results}: The estimated Prediction Intervals for the MSE-CV solution $\hat{\vtheta}_{MSE}$ after RPIE are now four times larger, meaning that the amplitude $\hat{\sigma}^2_{MSE}$ was largely underestimated. Table \ref{tab:Zhou_Results} also shows that the CPs for the testing set are close to their desired value $1-\alpha=90\%$.}

\begin{table*}[h!]
	\begin{center}
		\centering
		\caption{Performances of methods for \textit{Zhou} function (\ref{eq:Zhou}) in the second setting ($\hat{\sigma}^2_{\epsilon}=1.71 \ 10^{-2}$) ; here $1-\alpha=90\%$}
		\label{tab:Zhou_Results_Nugget}
		\begin{tabular}{ l c c c c c}
			\toprule &  \multicolumn{2}{c}{Before RPIE} & \multicolumn{2}{c}{After RPIE} & Full-Bayesian \\ 
			\toprule 
			\addlinespace 
			& MLE  & MSE-CV & MLE  & MSE-CV & - \\ 
			\midrule 
			\addlinespace 
			 $Q^2$ &  0.941 &  0.944  &  n.c  &  n.c &  0.941  \\ 
			\addlinespace 
			\addlinespace 
			$\mathbb{\tilde P}_{1-\alpha}$ & 99.4 &  100  &  90.0  &  90.0  & 99.3 \\
			$\text{CP}_{1-\alpha}$ & 99.3 &  100  &  92.0  &  85.3  &  99.6 \\
			\addlinespace 
			\addlinespace
			$\text{MPIW}_{1-\alpha}$ & $6.48\ 10^{-1}$ & $1.19$  &  $2.26 \ 10^{-1}$ &  $2.28 \ 10^{-1}$ & $6.56 \ 10^{-1}$ \\
			$\text{SdPIW}_{1-\alpha}$  & $6.88 \ 10^{-2}$ &  $2.56 \ 10^{-1}$ & $4.73 \ 10^{-2}$  & $5.27 \ 10^{-2}$  &  $6.97 \ 10^{-2}$ \\
			\addlinespace 
			\addlinespace
			Ct &  1min 20s  & 31min 22s  & 3min 39s  & 33min 37s & 4h 25min 59s \\
			\bottomrule
		\end{tabular}
    \end{center}
    \begin{tablenotes}[flushleft]
	    \scriptsize
	    \item $Q^2$: Accuracy; $\mathbb{\tilde P}_{1-\alpha}$: The Leave-One-Out CP in \%  on the training set; $\text{CP}_{1-\alpha}$: The CP in \% on the testing set; $\text{MPIW}$: Mean of Prediction Interval widths; $\text{SdPIW}$: standard deviation of Prediction Interval widths and Ct: computational time.
    \end{tablenotes}
\end{table*}

{In the second setting, the nugget effect is estimated to  $\hat{\sigma}^2_{\epsilon}=1.71 \ 10^{-2}$ by using cite{Iooss2017}. The results of MLE, MSE-CV and Full-Bayesian methods are shown in Table \ref{tab:Zhou_Results_Nugget}. The accuracy is still satisfying and similar to the previous setting, but the CP is close to $100\%$, meaning that the Prediction Intervals of all three methods are overestimated. Table \ref{tab:Zhou_Results_Nugget} shows that, with the RPIE method, we reduce Prediction Intervals width, five times shorter than Prediction Intervals of the MSE-CV solution, and three shorter than Prediction Intervals of the MLE solution. The variances of the obtained Prediction Intervals are between MLE and MSE-CV Prediction Intervals variances. One can notice also a decrease of $50\%$ of the MPIW compared to the first setting, while maintaining an optimal coverage of $1-\alpha=90\%$.}

{Example 3 illustrates a case of misspecified model without noise where the RPIE method adjusts Prediction Intervals width and improves the robustness of Prediction Intervals so that the CP is respected. One can also conclude that it is preferable to consider a nugget effect for shorter Prediction Intervals and optimal coverage}.

\subsection{Application to Gas production for future wells}
\label{Section:5.2}

{In this section, we illustrate the interest of the RPIE method in energy production forcasting. It includes many industrial applications such as battery capacity, wind turbine, solar panel performance or, more specifically, unconventional gas wells where a decline in production may be observed. We show that the RPIE can estimate robust Prediction Intervals, covering the lower bounds of level $\alpha/2=10\%$ (pessimistic scenario) and the upper bounds of level $1-\alpha/2 =90\%$ (optimistic scenario).}

{Indeed, a fundamental challenge of Oil and Gas companies is to predict their assets and their production capacities in the future. It drives both their exploration and development strategy. However, forecasting a well future production is challenging because subsurface reservoirs properties are never fully known. This makes estimating well production with their associated uncertainty a crucial task. The agencies \textbf{PRMS} and \textbf{SEC} \citep{SEC_ref,PRMS_ref} define specific rules \textbf{1P}/\textbf{2P}/\textbf{3P} for reserves estimates based on quantile estimates:
\begin{itemize}
    \item \textbf{1P}: 90\% of wells produce more than \textbf{1P} predictions \textbf{(proven)}.
    \item \textbf{2P}: 50\% of wells produce more than \textbf{2P} predictions \textbf{(probable)}.
    \item \textbf{3P}: 10\% of wells produce more than \textbf{3P} predictions \textbf{(possible)}. 
\end{itemize}}

{These rules are to be disclosed to security investors for publicly traded Oil and Gas companies and aim to provide investors with consistent information and associated value assessments. Many Machine Learning algorithms have shown their efficiency in estimating the median \textbf{2P} (e.g. using GP with MLE method, or MSE-CV if interested more in point-wise predictions) but failed to estimate \textbf{1P} and \textbf{3P}. Thus, the objective of this study is to build a proper estimation of the quantiles $p_{90\%}$ and $p_{10\%}$ by applying the RPIE method described in section \ref{sec4:PIforGP}.}

{Our dataset, \textit{field data}, is derived from unconventional wells localized in the \textit{Utica} shale reservoir, located in the north-east of the United States. It contains approximately $n=1850$ wells and $d=12$ variables, including localization, Cumulative Production of {natural gas over 12 months in MCFE}, completion design and exploitation conditions. The raw dataset can be found at the Ohio Oil \& Gas well locator of the Ohio Department of Natural Resources \citep{Ohiodnr}.}

	\begin{table*}[h!]
	\begin{center}
		\centering
		\caption{Results obtained for GP model, Random Forest and Gradient Boosting; here $1-\alpha=80\%$.}
		\label{tab:Production_Results}
		\begin{tabular}{ l c c c}
			\toprule 
			\addlinespace 
			& MLE & Random Forest & XGBoost \\ 
			\midrule 
			\addlinespace 
			$Q^2$   &  0.872  &  0.870  &  0.885  \\ 
			$\text{CP}_{1-\alpha}$  &  92.8  &  98.1  &  49.8 \\ 

			\addlinespace 
		    $\text{MPIW}_{1-\alpha}$ & 1.18  & 1.52 & 0.48 \\
		    $\text{SdPIW}_{1-\alpha}$ & 0.21  & 0.29 & 0.22\\
		    
		    \addlinespace 
		    Ct &  14min 37s   &  2s  &  1min 36s  \\
			\bottomrule
			\addlinespace 
		\end{tabular}
	\end{center}
	\begin{tablenotes}[flushleft]
	\scriptsize
	\item $Q^2$: Accuracy; $\text{CP}$: The CP in \% on validation set I; $\text{MPIW}$: Mean of Prediction Interval widths; $\text{SdPIW}$: standard deviation of Prediction Interval widths and Ct: computational time.
	\end{tablenotes}
\end{table*}


{We standardized the data $(\mathbf{X},\vy)$ and we divided into a $60\%-20\%-20\%$ partition of three datasets: a training set and two validation sets. The response $\vy$ (Cumulative Production over 12 months in MCFE) is noisy due to the uncertainty of the reservoir parameters in the field. The nugget effect $\sigma^2_{\epsilon}$ is unknown but estimated to $\widehat{\sigma}^2_{\epsilon}=0.16$ using the method of \cite{Iooss2017}.}


{Based on results drawn from the previous subsection and for practical reasons (particularly the computational cost of methods), we will present only the application of the RPIE method on the MLE solution. Table \ref{tab:Production_Results} shows the performances of the GP model trained by MLE compared with two other statistical models: Random Forest and Gradient Boosting whose Prediction Intervals are estimated using the Bootstrap method. Here we consider the  Prediction Intervals of level $1-\alpha=80\%$: the lower bound is the 10\% quantile ($p_{10\%}$) and the upper bound the 90\% quantile ($p_{90\%}$) of the predictive distribution.}

{The accuracy of the MLE model is $0.873$ and has approximately the same accuracy as other models like Random Forest or Gradient Boosting. Furthermore, the CP of the Prediction Intervals of $1−\alpha = 80\%$ is not satisfactory but it is quite \textit{reasonable} for MLE model compared to Random Forest (overestimated Prediction Intervals) or Gradient Boosting (underestimated Prediction Intervals). Finally, it appears that the GP model requires some computing resources to be built and to estimate its hyperparameters by MLE method.}

\begin{table*}[!ht]
	\begin{center}
		\centering
		\caption{Obtained results before and after RPIE method; here $1-\alpha=80\%$.}
		\label{tab:Mean_sd_CredibleIntervals}
		
		\begin{tabular}{ l  c  c }
			\toprule 
			& MLE before RPIE  & MLE after RPIE  \\ 
			\midrule 
			$\mathbb{\tilde P}_{1-\alpha}$ & 90.9  & 79.9 \\
			\addlinespace 
			$\text{CP}^{\mathrm{Val,1}}_{1-\alpha}$  &  92.6  &  81.0 \\ 
			$\text{MPIW}^{\mathrm{Val,1}}_{1-\alpha}$ & $1.18$  & $1.06$ \\
			$\text{SdPIW}^{\mathrm{Val,1}}_{1-\alpha}$  & $2.09 \ 10^{-1}$ &  $8.25 \ 10^{-3}$  \\
			\addlinespace 
			$\text{CP}^{\mathrm{Val,2}}_{1-\alpha}$  & 94.1  &  83.2\\
			$\text{MPIW}^{\mathrm{Val,2}}_{1-\alpha}$ & $1.17$  & $1.06$ \\
			$\text{SdPIW}^{\mathrm{Val,2}}_{1-\alpha}$  & $1.68 \ 10^{-1}$ &  $7.00 \ 10^{-3}$  \\
			\addlinespace 
			Ct  &  14min 37s   &  59min 25s \\
			\bottomrule
			\addlinespace 
		\end{tabular}
	\end{center}
	\begin{tablenotes}[flushleft]
		\scriptsize
		\item $\text{CP}^{\mathrm{Val,1}}_{1-\alpha}$ (resp. $\text{CP}^{\mathrm{Val,2}}_{1-\alpha}$) : The CP in \% on Validation set I (resp. Validation set II); $\text{MPIW}^{\mathrm{Val,1}}_{1-\alpha}$ (resp. $\text{MPIW}^{\mathrm{Val,2}}_{1-\alpha}$): Mean of Prediction Interval widths on Validation set I (resp. Validation set II); $\text{SdPIW}^{\mathrm{Val,1}}_{1-\alpha}$ (resp. $\text{SdPIW}^{\mathrm{Val,2}}_{1-\alpha}$): standard deviation of Prediction Interval widths on Validation set I (resp. Validation set II) and Ct: computational time.
	\end{tablenotes}
\end{table*}

{In the following, we define the MLE's solution as reference $\vtheta_0=\hat{\vtheta}_{ML}$ in the optimization problem (\ref{prob:relax_dWass}) for the quantiles $\alpha/2=10\%$ and $1-\alpha/2=90\%$ and we build robust Prediction Intervals confidence level $1-\alpha=80\%$ with the RPIE method. The results are presented in Table \ref{tab:Mean_sd_CredibleIntervals}. When considering the estimated Prediction Intervals by the RPIE method, we can see the CP is optimal for the training set and is close to $1-\alpha=80\%$ for both validation sets. Therefore, we fulfil the objective of estimating the upper and lower bounds, the obtained quantiles $p_{90\%}$ and $p_{10\%}$ respect \textbf{1P} and \textbf{3P} rules as mentioned above. 
Finally, in Figure \ref{fig:MLE_model}, we present the estimated Prediction Intervals defined by the upper bounds $P90$ and lower bounds $P10$ against the true values of $\vy$ on Validation set I. The x-axis designs well's indices ordered with respect to the barycenters of the Prediction Intervals (engineers choose this representation for interpretation purposes).
We can see that the estimated Prediction Intervals by the MLE method are not homogeneous, and some of them are longer. The RPIE method makes them shorter and more homogeneous as it can be seen in Figure \ref{fig:P50_model}, and in the evolution of the standard deviation width $\text{SdPIW}$ in Table \ref{tab:Mean_sd_CredibleIntervals}.}

{In a second attempt and following the engineers' recommendation, we consider a logarithmic transformation to the raw response $\vy$ to avoid having non-positive lower bounds and integrate heterogeneity between performant and less performant well. The accuracy of the MLE method decreases now to $Q^2=0.615$, the MLE method still over estimates Prediction Intervals as it can be seen in Table \ref{tab:Results_logProd}. Most claims of the previous analysis remain true, in particular we can clearly see (also in Figures \ref{fig:MLE_logmodel} and \ref{fig:P50_logmodel}) that Prediction Intervals obtained by RPIE are shorter and have reduced standard-deviations.}

\begin{figure}[ht!]
    \centering
    \begin{subfigure}[b]{0.45\textwidth}
        \centering
        \includegraphics[width=5.8cm,height=7cm]{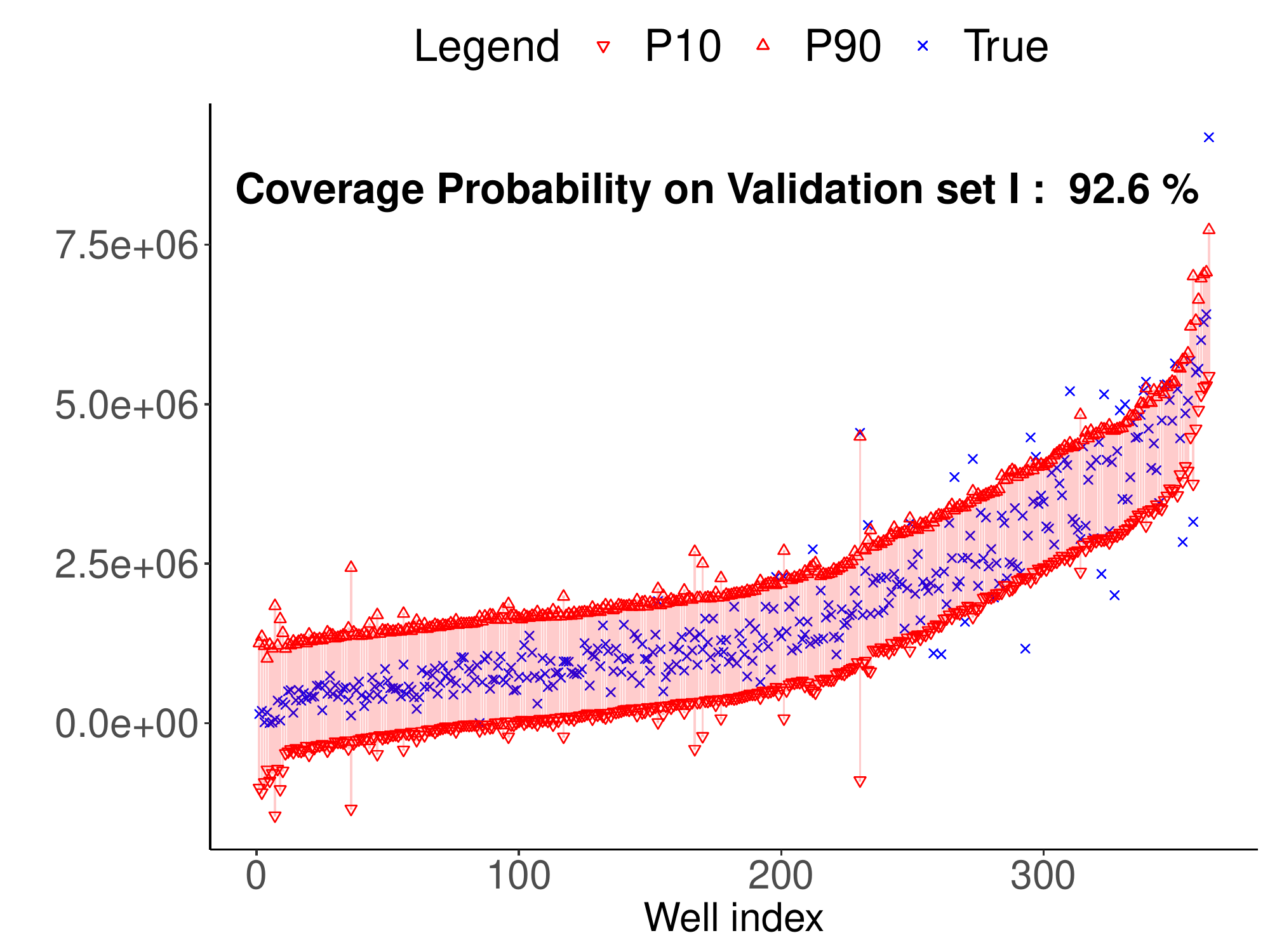}
        \caption{Before the RPIE on standardized output}
        \label{fig:MLE_model}
    \end{subfigure}
    \hfil
    \begin{subfigure}[b]{0.45\textwidth}
        \centering
        \includegraphics[width=5.8cm,height=7cm]{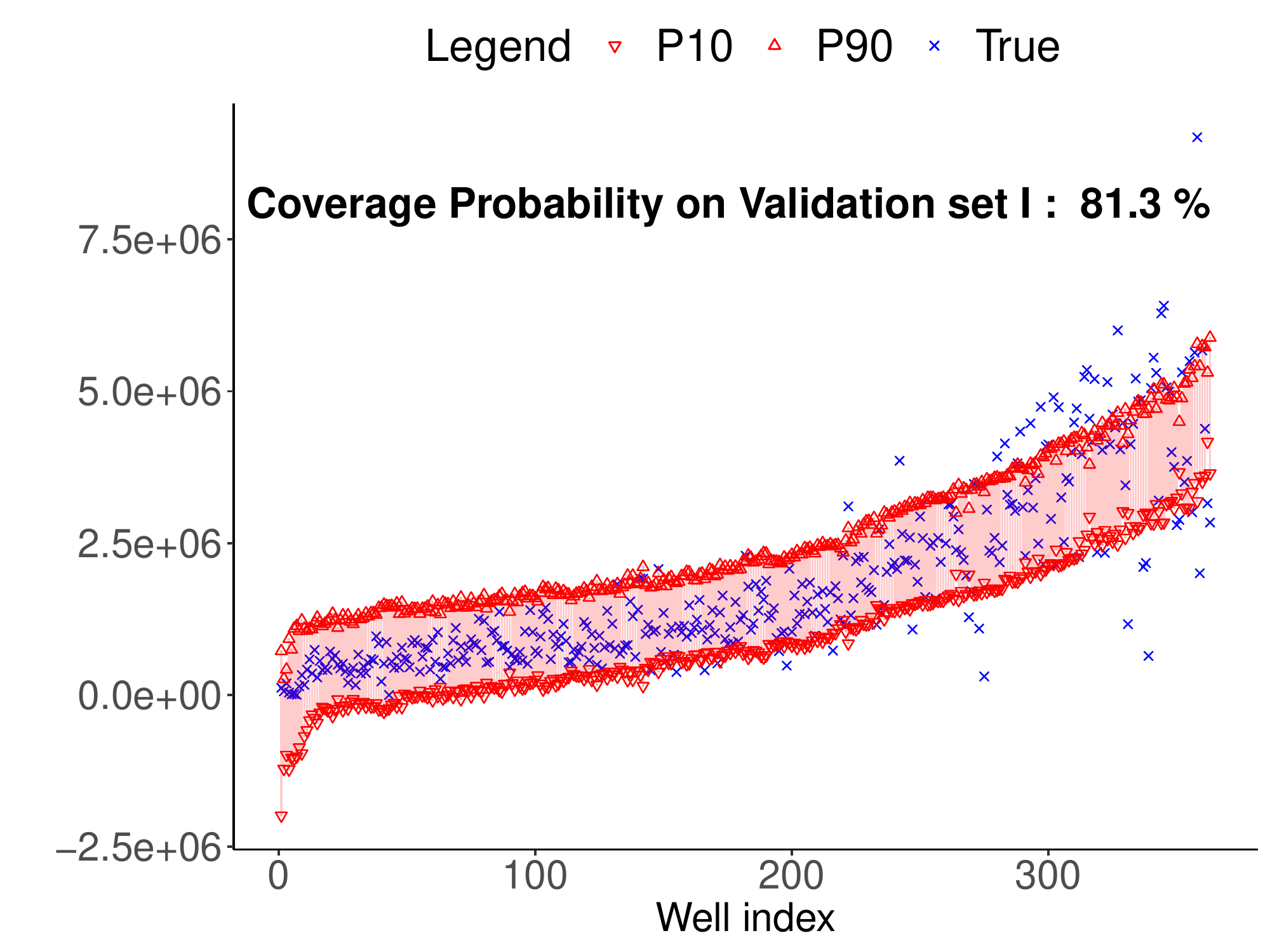}
        \caption{After the RPIE on standardized output}
        \label{fig:P50_model}
    \end{subfigure}
        \begin{subfigure}[b]{0.45\textwidth}
        \centering
        \includegraphics[width=5.8cm,height=7cm]{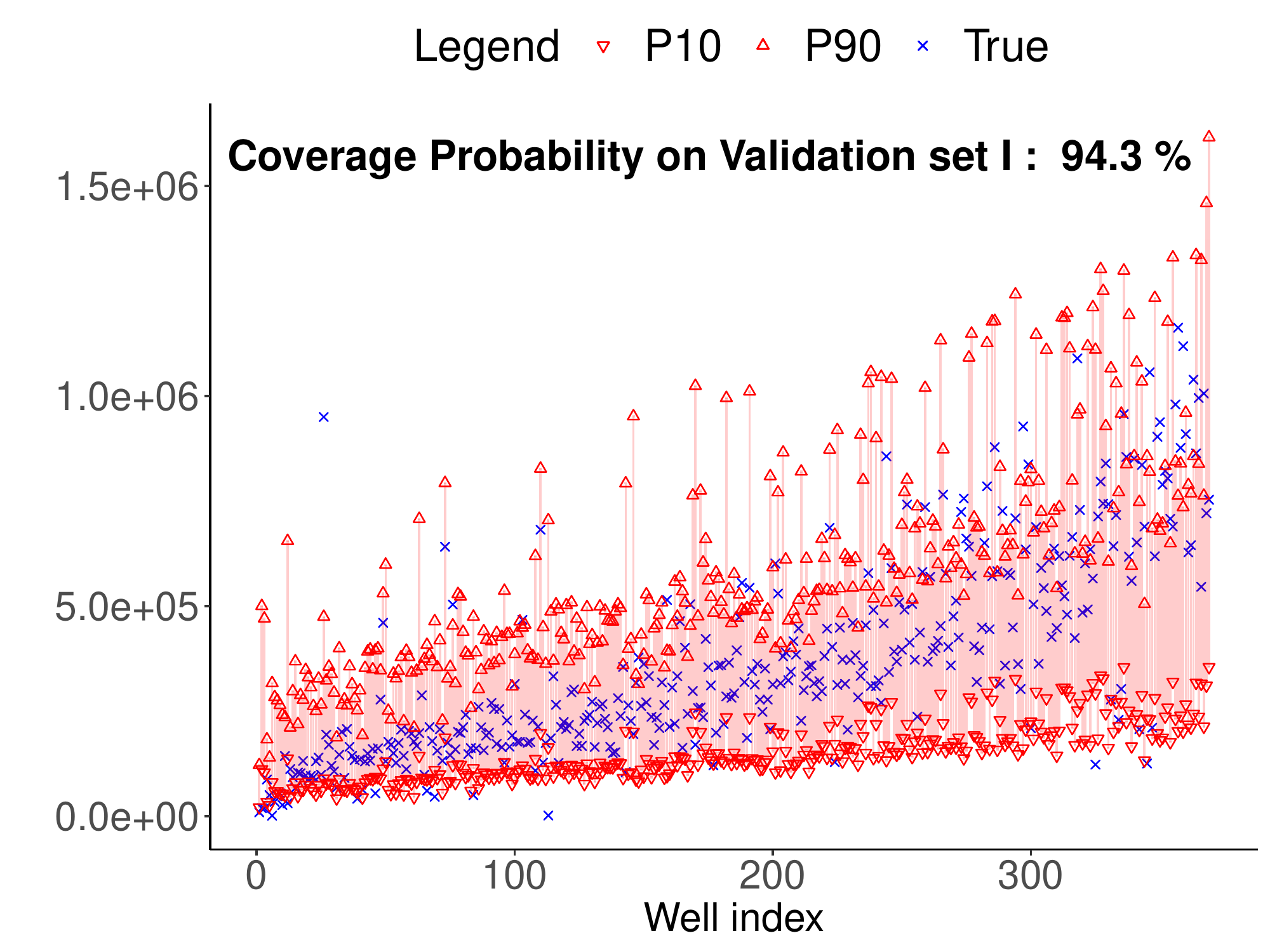}
        \caption{Before the RPIE on log output}
        \label{fig:MLE_logmodel}
    \end{subfigure}
    \hfil
    \begin{subfigure}[b]{0.47\textwidth}
        \centering
        \includegraphics[width=5.8cm,height=7cm]{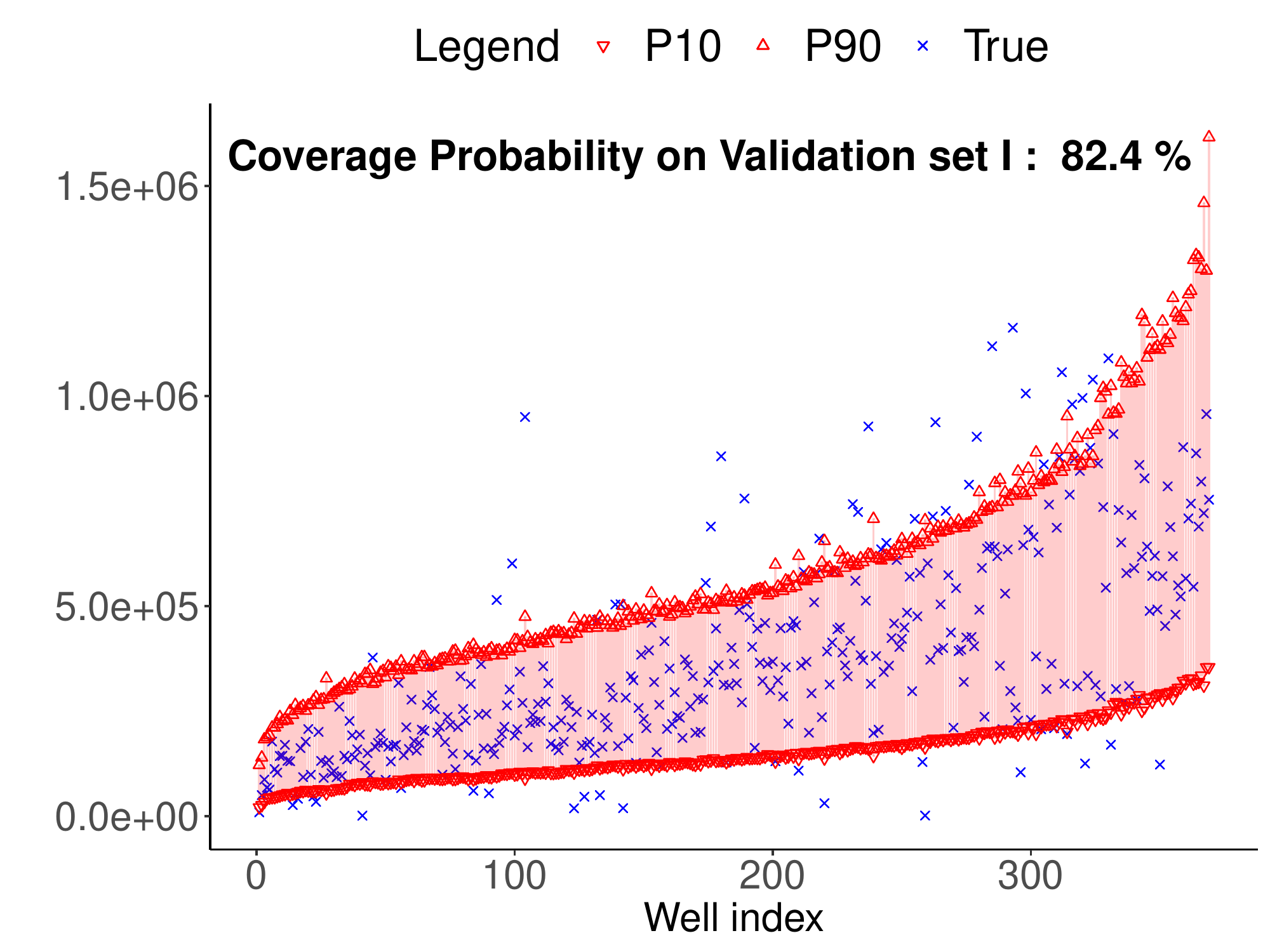}
        \caption{After the RPIE on log output}
        \label{fig:P50_logmodel}
    \end{subfigure}
    \caption{ Production data after re-scaling: True values vs $80\%$ confidence Prediction Intervals}
    \label{fig:Production_True_vs_Predicted}
\end{figure}

\begin{table*}[!ht]
	\begin{center}
		\centering
		\caption{Obtained results before and after RPIE method; $1-\alpha=80\%$. Here the output data are log-transformed.}
		\label{tab:Results_logProd}
		
		\begin{tabular}{ l  c  c }
			\toprule 
			& MLE before RPIE  & MLE after RPIE  \\ 
			\midrule 
			$\mathbb{\tilde P}_{1-\alpha}$ & 91.1  & 79.9 \\
			\addlinespace 
			$\text{CP}^{\mathrm{Val,1}}_{1-\alpha}$  &  94.3  &  83.2 \\ 
			$\text{MPIW}^{\mathrm{Val,1}}_{1-\alpha}$ & $1.53$  & $1.40$ \\
			$\text{SdPIW}^{\mathrm{Val,1}}_{1-\alpha}$  & $2.20 \ 10^{-1}$ &  $1.40 \ 10^{-2}$  \\
			\addlinespace 
			$\text{CP}^{\mathrm{Val,2}}_{1-\alpha}$  & 90.4  &  76.6\\
			$\text{MPIW}^{\mathrm{Val,2}}_{1-\alpha}$ & $1.54$  & $1.40$ \\
			$\text{SdPIW}^{\mathrm{Val,2}}_{1-\alpha}$  & $1.92 \ 10^{-1}$ &  $1.42 \ 10^{-2}$  \\
			\addlinespace 
			Ct  &  17min 47s   & 53min 21s \\
			\bottomrule
			\addlinespace 
		\end{tabular}
	\end{center}
	\begin{tablenotes}[flushleft]
		\scriptsize
		\item $\text{CP}^{\mathrm{Val,1}}_{1-\alpha}$ (resp. $\text{CP}^{\mathrm{Val,2}}_{1-\alpha}$) : The CP in \% on Validation set I (resp. Validation set II); $\text{MPIW}^{\mathrm{Val,1}}_{1-\alpha}$ (resp. $\text{MPIW}^{\mathrm{Val,2}}_{1-\alpha}$): Mean of Prediction Interval widths on Validation set I (resp. Validation set II); $\text{SdPIW}^{\mathrm{Val,1}}_{1-\alpha}$ (resp. $\text{SdPIW}^{\mathrm{Val,2}}_{1-\alpha}$): standard deviation of Prediction Interval widths on Validation set I (resp. Validation set II) and Ct: computational time.
	\end{tablenotes}
\end{table*}

\section{Conclusion}
\label{Conclusion}

In this paper, we have introduced a new approach for Prediction Intervals estimation based on the Cross-Validation method. We use the Gaussian Processes model because the predictive distribution at a new point is completely characterized by Gaussian {distribution}. We address an optimization problem for model’s hyperparameters estimation by considering the notion of Coverage Probability. The optimal hyperparameters are identified by minimizing the Wasserstein distance between the Gaussian distribution with the hyperparameters determined by Cross-Validation, and the Gaussian distribution with hyperparameters achieving the desired Coverage Probability. This method is relevant when the model is misspecified. It insures an optimal Leave-One-Out Coverage Probability for the training set. It also achieves a reasonable Coverage Probability for the validation set when it is available. It can be also extended to other statistical models with a predictive distribution but more detailed work is needed to consider the influence of hyperparameters on Prediction Interval's coverage and solve the optimization problem more efficiently in these cases. {Finally, it should be possible to include categorical inputs in the covariance function by using group kernels \citep{Roustant2020}, which would extend the application range of the RPIE method.}

\section*{Acknowledgements}

The author$^{(1)}$ would like to thank Achraf Ourir, Zinyat Agharzayeva (TotalEnergies SE - La D\'efense, France), Daniel Busby (TotalEnergies SE, CSTJF - Pau, France) and the research laboratory SINCLAIR (IA Commun Lab - Saclay, France) for useful discussions and helpful suggestions. This work was supported by TotalEnergies and the French National Agency for Research and Technology (ANRT).

\bibliography{mybibfile}

\begin{thebibliography}{58}
\expandafter\ifx\csname natexlab\endcsname\relax\def\natexlab#1{#1}\fi
\providecommand{\url}[1]{\texttt{#1}}
\providecommand{\href}[2]{#2}
\providecommand{\path}[1]{#1}
\providecommand{\DOIprefix}{doi:}
\providecommand{\ArXivprefix}{arXiv:}
\providecommand{\URLprefix}{URL: }
\providecommand{\Pubmedprefix}{pmid:}
\providecommand{\doi}[1]{\href{http://dx.doi.org/#1}{\path{#1}}}
\providecommand{\Pubmed}[1]{\href{pmid:#1}{\path{#1}}}
\providecommand{\bibinfo}[2]{#2}
\ifx\xfnm\relax \def\xfnm[#1]{\unskip,\space#1}\fi
\bibitem[{{Bachoc}(2013{\natexlab{a}})}]{Bachoc2013CrossVA}
\bibinfo{author}{{Bachoc}, F.} (\bibinfo{year}{2013}{\natexlab{a}}).
\newblock \bibinfo{title}{Cross validation and maximum likelihood estimations
  of hyper-parameters of {G}aussian processes with model misspecification}.
\newblock {\it \bibinfo{journal}{Computational Statistics \& Data Analysis}\/},
   {\it \bibinfo{volume}{66}\/}, \bibinfo{pages}{55--69}.
\bibitem[{{Bachoc}(2013{\natexlab{b}})}]{bachoc2013}
\bibinfo{author}{{Bachoc}, F.} (\bibinfo{year}{2013}{\natexlab{b}}).
\newblock {\it \bibinfo{title}{Estimation paramétrique de la fonction de
  covariance dans le modèle de Krigeage par processus Gaussiens : application
  à la quantification des incertitues en simulation numérique}\/}.
\newblock Ph.D. thesis University Paris 7.
\newblock \URLprefix \url{http://www.theses.fr/2013PA077111}.
\bibitem[{{Berge}(1963)}]{Berge1963}
\bibinfo{author}{{Berge}, C.} (\bibinfo{year}{1963}).
\newblock {\it \bibinfo{title}{Topological Spaces: Including a Treatment of
  Multi-valued Functions, Vector Spaces and Convexity}\/}.
\newblock \bibinfo{publisher}{Oliver \& Boyd}.
\bibitem[{{Currin} et~al.(1991){Currin}, {Mitchell}, {Morris} \&
  {Ylvisaker}}]{Currin1991BayesianPO}
\bibinfo{author}{{Currin}, C.}, \bibinfo{author}{{Mitchell}, T.~J.},
  \bibinfo{author}{{Morris}, M.~D.}, \& \bibinfo{author}{{Ylvisaker}, D.}
  (\bibinfo{year}{1991}).
\newblock \bibinfo{title}{Bayesian prediction of deterministic functions, with
  applications to the design and analysis of computer experiments}.
\newblock {\it \bibinfo{journal}{Journal of the American Statistical
  Association}\/},  {\it \bibinfo{volume}{86}\/}, \bibinfo{pages}{953--963}.
\bibitem[{{De Oliveira}(2007)}]{DeOliveira2007}
\bibinfo{author}{{De Oliveira}, V.} (\bibinfo{year}{2007}).
\newblock \bibinfo{title}{Objective {B}ayesian analysis of spatial data with
  measurement error}.
\newblock {\it \bibinfo{journal}{The Canadian Journal of Statistics / La Revue
  Canadienne de Statistique}\/},  {\it \bibinfo{volume}{35}\/},
  \bibinfo{pages}{283--301}.
\bibitem[{{De Oliveira} et~al.(2001){De Oliveira}, {Berger} \&
  {Sansó}}]{Berger2001}
\bibinfo{author}{{De Oliveira}, V.}, \bibinfo{author}{{Berger}, J.~O.}, \&
  \bibinfo{author}{{Sansó}, B.} (\bibinfo{year}{2001}).
\newblock \bibinfo{title}{Objective {B}ayesian analysis of spatially correlated
  data}.
\newblock {\it \bibinfo{journal}{Journal of the American Statistical
  Association}\/},  {\it \bibinfo{volume}{96}\/}, \bibinfo{pages}{1361--1374}.
\bibitem[{Dubrule(1983)}]{Dubrule1983}
\bibinfo{author}{Dubrule, O.} (\bibinfo{year}{1983}).
\newblock \bibinfo{title}{Cross validation of kriging in a unique
  neighborhood}.
\newblock {\it \bibinfo{journal}{Journal of the International Association for
  Mathematical Geology}\/},  {\it \bibinfo{volume}{15}\/},
  \bibinfo{pages}{687--699}.
\bibitem[{Efron(1992)}]{EfronJackknife_After_Bootstrap}
\bibinfo{author}{Efron, B.} (\bibinfo{year}{1992}).
\newblock \bibinfo{title}{Jackknife-after-bootstrap standard errors and
  influence functions}.
\newblock {\it \bibinfo{journal}{Journal of the Royal Statistical Society.
  Series B (Methodological)}\/},  {\it \bibinfo{volume}{54}\/},
  \bibinfo{pages}{83--127}. \URLprefix
  \url{http://www.jstor.org/stable/2345949}.
\bibitem[{{Efron} \& {Tibshirani}(1994)}]{efron1994Bootstrap}
\bibinfo{author}{{Efron}, B.}, \& \bibinfo{author}{{Tibshirani}, R.}
  (\bibinfo{year}{1994}).
\newblock {\it \bibinfo{title}{An Introduction to the Bootstrap}\/}.
\newblock Chapman \& Hall/CRC Monographs on Statistics \& Applied Probability.
\newblock \bibinfo{publisher}{Taylor \& Francis}.
\newblock \URLprefix \url{https://books.google.fr/books?id=gLlpIUxRntoC}.
\bibitem[{Filippone et~al.(2013)Filippone, Zhong \& Girolami}]{FilipponeZG13}
\bibinfo{author}{Filippone, M.}, \bibinfo{author}{Zhong, M.}, \&
  \bibinfo{author}{Girolami, M.~A.} (\bibinfo{year}{2013}).
\newblock \bibinfo{title}{A comparative evaluation of stochastic-based
  inference methods for {G}aussian process models}.
\newblock {\it \bibinfo{journal}{Maching Learning}\/},  {\it
  \bibinfo{volume}{93}\/}, \bibinfo{pages}{93--114}. \URLprefix
  \url{https://doi.org/10.1007/s10994-013-5388-x}.
  \DOIprefix\doi{10.1007/s10994-013-5388-x}.
\bibitem[{Forrester et~al.(2008)Forrester, S{\'o}bester \&
  Keane}]{Forrester2008}
\bibinfo{author}{Forrester, A. I.~J.}, \bibinfo{author}{S{\'o}bester, A.}, \&
  \bibinfo{author}{Keane, A.~J.} (\bibinfo{year}{2008}).
\newblock {\it \bibinfo{title}{Engineering Design Via Surrogate Modelling: A
  Practical Guide}\/}.
\newblock Progress in Astronautics and Aeronautics.
\newblock \bibinfo{publisher}{American Institute of Aeronautics and
  Astronautics}.
\bibitem[{{Gal} \& {Ghahramani}(2016)}]{GhahramaniPMLR}
\bibinfo{author}{{Gal}, Y.}, \& \bibinfo{author}{{Ghahramani}, Z.}
  (\bibinfo{year}{2016}).
\newblock \bibinfo{title}{Dropout as a {B}ayesian approximation: Representing
  model uncertainty in deep learning}.
\newblock In {\it \bibinfo{booktitle}{Proceedings of The 33rd International
  Conference on Machine Learning}\/} (pp. \bibinfo{pages}{1050--1059}).
\newblock \bibinfo{address}{New York, New York, USA}: \bibinfo{publisher}{PMLR}
  volume~\bibinfo{volume}{48} of {\it \bibinfo{series}{Proceedings of Machine
  Learning Research}\/}.
\bibitem[{{Hastie} et~al.(2009){Hastie}, {Tibshirani} \&
  {Friedman}}]{Tibshirani2009}
\bibinfo{author}{{Hastie}, T.}, \bibinfo{author}{{Tibshirani}, R.}, \&
  \bibinfo{author}{{Friedman}, J.} (\bibinfo{year}{2009}).
\newblock {\it \bibinfo{title}{The Elements of Statistical Learning: Data
  Mining, Inference, and Prediction, Second Edition}\/}.
\newblock Springer Series in Statistics.
\newblock \bibinfo{publisher}{Springer New York}.
\bibitem[{{Heskes}(1997)}]{Heskes97practicalconfidence}
\bibinfo{author}{{Heskes}, T.} (\bibinfo{year}{1997}).
\newblock \bibinfo{title}{Practical confidence and prediction intervals}.
\newblock In {\it \bibinfo{booktitle}{Advances in Neural Information Processing
  Systems 9}\/} (pp. \bibinfo{pages}{176--182}).
\newblock \bibinfo{publisher}{MIT press}.
\bibitem[{{Hong} et~al.(2009){Hong}, {Meeker} \& {McCalley}}]{Hong2009}
\bibinfo{author}{{Hong}, Y.}, \bibinfo{author}{{Meeker}, W.~Q.}, \&
  \bibinfo{author}{{McCalley}, J.~D.} (\bibinfo{year}{2009}).
\newblock \bibinfo{title}{Prediction of remaining life of power transformers
  based on left truncated and right censored lifetime data}.
\newblock {\it \bibinfo{journal}{Ann. Appl. Stat.}\/},  {\it
  \bibinfo{volume}{3}\/}, \bibinfo{pages}{857--879}.
  \DOIprefix\doi{10.1214/00-AOAS231}.
\bibitem[{{Hwang} \& {Ding}(1997)}]{DeltaNNet97}
\bibinfo{author}{{Hwang}, J. T.~G.}, \& \bibinfo{author}{{Ding}, A.~A.}
  (\bibinfo{year}{1997}).
\newblock \bibinfo{title}{Prediction intervals for artificial neural networks}.
\newblock {\it \bibinfo{journal}{Journal of the American Statistical
  Association}\/},  {\it \bibinfo{volume}{92}\/}, \bibinfo{pages}{748--757}.
\bibitem[{{Iooss} \& {Marrel}(2017)}]{Iooss2017}
\bibinfo{author}{{Iooss}, B.}, \& \bibinfo{author}{{Marrel}, A.}
  (\bibinfo{year}{2017}).
\newblock \bibinfo{title}{{An efficient methodology for the analysis and
  modeling of computer experiments with large number of inputs}}.
\newblock In {\it \bibinfo{booktitle}{{UNCECOMP 2017 2nd ECCOMAS Thematic
  Conference on Uncertainty Quantification in Computational Sciences and
  Engineering}}\/} (pp. \bibinfo{pages}{187--197}).
\newblock \bibinfo{address}{Rhodes Island, Greece}.
\bibitem[{{Khosravi} et~al.(2010){Khosravi}, {Nahavandi} \&
  {Creighton}}]{KhosraviMPIW}
\bibinfo{author}{{Khosravi}, A.}, \bibinfo{author}{{Nahavandi}, S.}, \&
  \bibinfo{author}{{Creighton}, D.} (\bibinfo{year}{2010}).
\newblock \bibinfo{title}{A prediction interval-based approach to determine
  optimal structures of neural network metamodels}.
\newblock {\it \bibinfo{journal}{Expert Syst. Appl.}\/},  {\it
  \bibinfo{volume}{37}\/}, \bibinfo{pages}{2377--2387}.
  \DOIprefix\doi{10.1016/j.eswa.2009.07.059}.
\bibitem[{{Khosravi} et~al.(2011){Khosravi}, {Nahavandi}, {Creighton} \&
  {Atiya}}]{LUBE_2011}
\bibinfo{author}{{Khosravi}, A.}, \bibinfo{author}{{Nahavandi}, S.},
  \bibinfo{author}{{Creighton}, D.}, \& \bibinfo{author}{{Atiya}, A.~F.}
  (\bibinfo{year}{2011}).
\newblock \bibinfo{title}{Lower upper bound estimation method for construction
  of neural network-based prediction intervals}.
\newblock {\it \bibinfo{journal}{IEEE Transactions on Neural Networks}\/},
  {\it \bibinfo{volume}{22}\/}, \bibinfo{pages}{337--346}.
\bibitem[{{Kleijnen} \& {Sargent}(2000)}]{Kleijnen2000AMF}
\bibinfo{author}{{Kleijnen}, J. P.~C.}, \& \bibinfo{author}{{Sargent}, R.~G.}
  (\bibinfo{year}{2000}).
\newblock \bibinfo{title}{A methodology for fitting and validating metamodels
  in simulation}.
\newblock {\it \bibinfo{journal}{European Journal of Operational Research}\/},
  {\it \bibinfo{volume}{120}\/}, \bibinfo{pages}{14--29}.
\bibitem[{{Landon} \& {Singpurwalla}(2008)}]{PICP_2008}
\bibinfo{author}{{Landon}, J.}, \& \bibinfo{author}{{Singpurwalla}, N.}
  (\bibinfo{year}{2008}).
\newblock \bibinfo{title}{Choosing a coverage probability for prediction
  intervals}.
\newblock {\it \bibinfo{journal}{The American Statistician}\/},  {\it
  \bibinfo{volume}{62}\/}, \bibinfo{pages}{120--124}.
  \DOIprefix\doi{10.1198/000313008X304062}.
\bibitem[{{Lawless} \& {Fredette}(2005)}]{LawlessPI_2005}
\bibinfo{author}{{Lawless}, J.~F.}, \& \bibinfo{author}{{Fredette}, M.}
  (\bibinfo{year}{2005}).
\newblock \bibinfo{title}{Frequentist prediction intervals and predictive
  distributions}.
\newblock {\it \bibinfo{journal}{Biometrika}\/},  {\it \bibinfo{volume}{92}\/},
  \bibinfo{pages}{529--542}.
\bibitem[{{Lei} et~al.(2018){Lei}, {G’Sell}, {Rinaldo}, {Tibshirani} \&
  {Wasserman}}]{Lei_SCIntervals2018}
\bibinfo{author}{{Lei}, J.}, \bibinfo{author}{{G’Sell}, M.},
  \bibinfo{author}{{Rinaldo}, A.}, \bibinfo{author}{{Tibshirani}, R.~J.}, \&
  \bibinfo{author}{{Wasserman}, L.} (\bibinfo{year}{2018}).
\newblock \bibinfo{title}{Distribution-free predictive inference for
  regression}.
\newblock {\it \bibinfo{journal}{Journal of the American Statistical
  Association}\/},  {\it \bibinfo{volume}{113}\/}, \bibinfo{pages}{1094--1111}.
  \DOIprefix\doi{10.1080/01621459.2017.1307116}.
\bibitem[{{Li} et~al.(2018){Li}, {Wang}, {Lei}, {Zhang}, {Liu} \&
  {Zheng}}]{LI201897}
\bibinfo{author}{{Li}, K.}, \bibinfo{author}{{Wang}, R.},
  \bibinfo{author}{{Lei}, H.}, \bibinfo{author}{{Zhang}, T.},
  \bibinfo{author}{{Liu}, Y.}, \& \bibinfo{author}{{Zheng}, X.}
  (\bibinfo{year}{2018}).
\newblock \bibinfo{title}{Interval prediction of solar power using an improved
  bootstrap method}.
\newblock {\it \bibinfo{journal}{Solar Energy}\/},  {\it
  \bibinfo{volume}{159}\/}, \bibinfo{pages}{97 -- 112}. \URLprefix
  \url{http://www.sciencedirect.com/science/article/pii/S0038092X17309313}.
  \DOIprefix\doi{https://doi.org/10.1016/j.solener.2017.10.051}.
\bibitem[{MacKay(1992)}]{MacKay_BNnet1992}
\bibinfo{author}{MacKay, D. J.~C.} (\bibinfo{year}{1992}).
\newblock \bibinfo{title}{A practical {B}ayesian framework for backpropagation
  networks}.
\newblock {\it \bibinfo{journal}{Neural Comput.}\/},  {\it
  \bibinfo{volume}{4}\/}, \bibinfo{pages}{448–472}.
  \DOIprefix\doi{10.1162/neco.1992.4.3.448}.
\bibitem[{Mardia \& Marshall(1984)}]{MardiaMarshall84}
\bibinfo{author}{Mardia, K.~V.}, \& \bibinfo{author}{Marshall, R.~J.}
  (\bibinfo{year}{1984}).
\newblock \bibinfo{title}{Maximum likelihood estimation of models for residual
  covariance in spatial regression}.
\newblock {\it \bibinfo{journal}{Biometrika}\/},  {\it \bibinfo{volume}{71}\/},
  \bibinfo{pages}{135--146}.
\bibitem[{{Meinshausen}(2006)}]{QRF2006}
\bibinfo{author}{{Meinshausen}, N.} (\bibinfo{year}{2006}).
\newblock \bibinfo{title}{Quantile regression forests}.
\newblock {\it \bibinfo{journal}{Journal of Machine Learning Research}\/},
  {\it \bibinfo{volume}{7}\/}, \bibinfo{pages}{983–999}.
\bibitem[{{Moon}(2010)}]{Moon2010}
\bibinfo{author}{{Moon}, H.} (\bibinfo{year}{2010}).
\newblock {\it \bibinfo{title}{Design and Analysis of Computer Experiments for
  Screening Input Variables}\/}.
\newblock Ph.D. thesis The Ohio State University.
\bibitem[{{Morokoff} \& {Caflisch}(1995)}]{Morokoff95}
\bibinfo{author}{{Morokoff}, W.~J.}, \& \bibinfo{author}{{Caflisch}, R.~E.}
  (\bibinfo{year}{1995}).
\newblock \bibinfo{title}{Quasi-monte carlo integration}.
\newblock {\it \bibinfo{journal}{Journal of computational physics}\/},  {\it
  \bibinfo{volume}{122}\/}, \bibinfo{pages}{218--230}.
\bibitem[{{Mur\'e}(2018)}]{mure2018}
\bibinfo{author}{{Mur\'e}, J.} (\bibinfo{year}{2018}).
\newblock {\it \bibinfo{title}{Objective Bayesian analysis of Kriging models
  with anisotropic correlation kernel}\/}.
\newblock Ph.D. thesis Sorbonne Paris Cité.
\newblock \URLprefix \url{http://www.theses.fr/2018USPCC069}.
\bibitem[{Muré(2021)}]{mure2020propriety}
\bibinfo{author}{Muré, J.} (\bibinfo{year}{2021}).
\newblock \bibinfo{title}{{Propriety of the reference posterior distribution in
  {G}aussian process modeling}}.
\newblock {\it \bibinfo{journal}{The Annals of Statistics}\/},  {\it
  \bibinfo{volume}{49}\/}, \bibinfo{pages}{2356 -- 2377}. \URLprefix
  \url{https://doi.org/10.1214/20-AOS2040}. \DOIprefix\doi{10.1214/20-AOS2040}.
\bibitem[{Neal(1993)}]{Neal93}
\bibinfo{author}{Neal, R.~M.} (\bibinfo{year}{1993}).
\newblock {\it \bibinfo{title}{Probabilistic Inference Using Markov Chain Monte
  Carlo Methods}\/}.
\newblock \bibinfo{type}{Technical Report} Dept. of Computer Science,
  University of Toronto.
\bibitem[{Neal(1996)}]{Neal_1996}
\bibinfo{author}{Neal, R.~M.} (\bibinfo{year}{1996}).
\newblock {\it \bibinfo{title}{Bayesian Learning for Neural Networks}\/}.
\newblock \bibinfo{publisher}{Springer New York}.
\newblock \DOIprefix\doi{10.1007/978-1-4612-0745-0}.
\bibitem[{{Nix} \& {Weigend}(1994)}]{Nix_MVENNET}
\bibinfo{author}{{Nix}, D.~A.}, \& \bibinfo{author}{{Weigend}, A.~S.}
  (\bibinfo{year}{1994}).
\newblock \bibinfo{title}{Estimating the mean and variance of the target
  probability distribution}.
\newblock In {\it \bibinfo{booktitle}{Proceedings of 1994 IEEE International
  Conference on Neural Networks (ICNN'94)}\/} (pp. \bibinfo{pages}{55--60
  vol.1}).
\newblock volume~\bibinfo{volume}{1}.
\bibitem[{Oakley \& O'Hagan(2004)}]{OHagan2004ProbabilisticSA}
\bibinfo{author}{Oakley, J.~E.}, \& \bibinfo{author}{O'Hagan, A.}
  (\bibinfo{year}{2004}).
\newblock \bibinfo{title}{Probabilistic sensitivity analysis of complex models:
  A {B}ayesian approach}.
\newblock {\it \bibinfo{journal}{Journal of the Royal Statistical Society:
  Series B (Statistical Methodology)}\/},  {\it \bibinfo{volume}{66}\/},
  \bibinfo{pages}{751 -- 769}.
\bibitem[{{Ohio Department of Natural Resources}(2022)}]{Ohiodnr}
\bibinfo{author}{{Ohio Department of Natural Resources}}
  (\bibinfo{year}{2022}).
\newblock \bibinfo{title}{The {O}hio oil and gas well locator}.
\newblock
  \bibinfo{note}{\url{https://ohiodnr.gov/wps/portal/gov/odnr/discover-and-learn/safety-conservation/about-odnr/oil-gas/oil-gas-resources/well-locator}}.
\bibitem[{{Pang} et~al.(2018){Pang}, {Liu}, {Peng} \& {Peng}}]{Pang2018}
\bibinfo{author}{{Pang}, J.}, \bibinfo{author}{{Liu}, D.},
  \bibinfo{author}{{Peng}, Y.}, \& \bibinfo{author}{{Peng}, X.}
  (\bibinfo{year}{2018}).
\newblock \bibinfo{title}{Optimize the coverage probability of prediction
  interval for anomaly detection of sensor-based monitoring series}.
\newblock {\it \bibinfo{journal}{Sensors (Basel, Switzerland)}\/},  {\it
  \bibinfo{volume}{18}\/}. \DOIprefix\doi{10.3390/s18040967}.
\bibitem[{{Pearce} et~al.(2018){Pearce}, {Brintrup}, {Zaki} \&
  {Neely}}]{Pearce18a}
\bibinfo{author}{{Pearce}, T.}, \bibinfo{author}{{Brintrup}, A.},
  \bibinfo{author}{{Zaki}, M.}, \& \bibinfo{author}{{Neely}, A.}
  (\bibinfo{year}{2018}).
\newblock \bibinfo{title}{High-quality prediction intervals for deep learning:
  A distribution-free, ensembled approach}.
\newblock In {\it \bibinfo{booktitle}{Proceedings of the 35th International
  Conference on Machine Learning}\/} (pp. \bibinfo{pages}{4075--4084}).
\newblock \bibinfo{address}{Stockholm Sweden}: \bibinfo{publisher}{PMLR}
  volume~\bibinfo{volume}{80} of {\it \bibinfo{series}{Proceedings of Machine
  Learning Research}\/}.
\bibitem[{{Rasmussen} \& {Williams}(2005)}]{Rasmussen2006}
\bibinfo{author}{{Rasmussen}, C.~E.}, \& \bibinfo{author}{{Williams}, C. K.~I.}
  (\bibinfo{year}{2005}).
\newblock {\it \bibinfo{title}{Gaussian Processes for Machine Learning
  (Adaptive Computation and Machine Learning)}\/}.
\newblock \bibinfo{publisher}{The MIT Press}.
\bibitem[{Ren et~al.(2012)Ren, Sun \& He}]{Ren2012ObjectiveBA}
\bibinfo{author}{Ren, C.}, \bibinfo{author}{Sun, D.}, \& \bibinfo{author}{He,
  C.~Z.} (\bibinfo{year}{2012}).
\newblock \bibinfo{title}{Objective bayesian analysis for a spatial model with
  nugget effects}.
\newblock {\it \bibinfo{journal}{Journal of Statistical Planning and
  Inference}\/},  {\it \bibinfo{volume}{142}\/}, \bibinfo{pages}{1933--1946}.
\bibitem[{Robert \& Casella(2004)}]{Robert_2004}
\bibinfo{author}{Robert, C.~P.}, \& \bibinfo{author}{Casella, G.}
  (\bibinfo{year}{2004}).
\newblock {\it \bibinfo{title}{Monte Carlo Statistical Methods}\/}.
\newblock \bibinfo{publisher}{Springer New York}.
\newblock \DOIprefix\doi{10.1007/978-1-4757-4145-2}.
\bibitem[{Romano et~al.(2019)Romano, Patterson \& Cand{è}s}]{Candes_CI}
\bibinfo{author}{Romano, Y.}, \bibinfo{author}{Patterson, E.}, \&
  \bibinfo{author}{Cand{è}s, E.~J.} (\bibinfo{year}{2019}).
\newblock \bibinfo{title}{Conformalized quantile regression}.
\newblock In {\it \bibinfo{booktitle}{Advances in Neural Information Processing
  Systems 32 (NIPS 2019)}\/} (pp. \bibinfo{pages}{3538--3548}).
\newblock \bibinfo{publisher}{Curran Associates, Inc.}
\newblock \URLprefix
  \url{papers.nips.cc/paper/8613-conformalized-quantile-regression}.
\bibitem[{{Rosenblatt}(1989)}]{Yaglom1989}
\bibinfo{author}{{Rosenblatt}, M.} (\bibinfo{year}{1989}).
\newblock \bibinfo{title}{Review: {A}. {M}. {Y}aglom, correlation theory of
  stationary and random functions vol. i; basic results, vol. ii, supplementary
  notes and references}.
\newblock {\it \bibinfo{journal}{Bulletin (New Series) of the American
  Mathematical Society}\/},  {\it \bibinfo{volume}{20}\/},
  \bibinfo{pages}{207--211}. \URLprefix
  \url{https://projecteuclid.org:443/euclid.bams/1183555023}.
\bibitem[{Roustant et~al.(2020)Roustant, Padonou, Deville, Clément, Perrin,
  Giorla \& Wynn}]{Roustant2020}
\bibinfo{author}{Roustant, O.}, \bibinfo{author}{Padonou, E.},
  \bibinfo{author}{Deville, Y.}, \bibinfo{author}{Clément, A.},
  \bibinfo{author}{Perrin, G.}, \bibinfo{author}{Giorla, J.}, \&
  \bibinfo{author}{Wynn, H.} (\bibinfo{year}{2020}).
\newblock \bibinfo{title}{Group kernels for {G}aussian process metamodels with
  categorical inputs}.
\newblock {\it \bibinfo{journal}{SIAM/ASA Journal on Uncertainty
  Quantification}\/},  {\it \bibinfo{volume}{8}\/}, \bibinfo{pages}{775--806}.
  \URLprefix \url{https://doi.org/10.1137/18M1209386}.
  \DOIprefix\doi{10.1137/18M1209386}.
\bibitem[{Santner et~al.(2003)Santner, Williams \& Notz}]{Santner_2003}
\bibinfo{author}{Santner, T.~J.}, \bibinfo{author}{Williams, B.~J.}, \&
  \bibinfo{author}{Notz, W.~I.} (\bibinfo{year}{2003}).
\newblock {\it \bibinfo{title}{The Design and Analysis of Computer
  Experiments}\/}.
\newblock \bibinfo{publisher}{Springer New York}.
\newblock \DOIprefix\doi{10.1007/978-1-4757-3799-8}.
\bibitem[{Securities \& Commission(2010)}]{PRMS_ref}
\bibinfo{author}{Securities}, \& \bibinfo{author}{Commission, E.}
  (\bibinfo{year}{2010}).
\newblock \bibinfo{title}{Modernization of oil and gas reporting, revisions and
  additions to the definition section in rule 4-10 of regulation s-x}.
\newblock
  \bibinfo{note}{\url{https://www.sec.gov/rules/final/2008/33-8995.pdf}}.
\bibitem[{Shapiro \& Wilk(1965)}]{Shapiro65}
\bibinfo{author}{Shapiro, S.~S.}, \& \bibinfo{author}{Wilk, M.~B.}
  (\bibinfo{year}{1965}).
\newblock \bibinfo{title}{An analysis of variance test for normality (complete
  samples)}.
\newblock {\it \bibinfo{journal}{Biometrika}\/},  {\it \bibinfo{volume}{52}\/},
  \bibinfo{pages}{591--611}. \URLprefix
  \url{http://www.jstor.org/stable/2333709}.
\bibitem[{{Society of Petroleum Engineers}(2022)}]{SEC_ref}
\bibinfo{author}{{Society of Petroleum Engineers}} (\bibinfo{year}{2022}).
\newblock \bibinfo{title}{Petroleum reserves and resources definitions}.
\newblock \bibinfo{note}{\url{https://www.spe.org/en/industry/reserves/}}.
\bibitem[{Stein(1999)}]{Stein_1999}
\bibinfo{author}{Stein, M.~L.} (\bibinfo{year}{1999}).
\newblock {\it \bibinfo{title}{Interpolation of Spatial Data}\/}.
\newblock \bibinfo{publisher}{Springer New York}.
\newblock \DOIprefix\doi{10.1007/978-1-4612-1494-6}.
\bibitem[{Tipping(2004)}]{Tipping2004}
\bibinfo{author}{Tipping, M.~E.} (\bibinfo{year}{2004}).
\newblock \bibinfo{title}{Bayesian inference: An introduction to principles and
  practice in machine learning}.
\newblock In {\it \bibinfo{booktitle}{Advanced Lectures on Machine Learning: ML
  Summer Schools 2003, Canberra, Australia, February 2 - 14, 2003,
  T{\"u}bingen, Germany, August 4 - 16, 2003, Revised Lectures}\/} (pp.
  \bibinfo{pages}{41--62}).
\newblock \bibinfo{address}{Berlin, Heidelberg}: \bibinfo{publisher}{Springer
  Berlin Heidelberg}.
\bibitem[{{Villani}(2009)}]{Villani2009}
\bibinfo{author}{{Villani}, C.} (\bibinfo{year}{2009}).
\newblock \bibinfo{title}{The {W}asserstein distances}.
\newblock In {\it \bibinfo{booktitle}{Optimal Transport: Old and New}\/} (pp.
  \bibinfo{pages}{93--111}).
\newblock \bibinfo{address}{Berlin, Heidelberg}: \bibinfo{publisher}{Springer
  Berlin Heidelberg}.
\bibitem[{{Wager} et~al.(2014){Wager}, {Hastie} \&
  {Efron}}]{Wager_confidenceintervals}
\bibinfo{author}{{Wager}, S.}, \bibinfo{author}{{Hastie}, T.}, \&
  \bibinfo{author}{{Efron}, B.} (\bibinfo{year}{2014}).
\newblock \bibinfo{title}{Confidence intervals for random forests: the
  jackknife and the infinitesimal jackknife}.
\newblock {\it \bibinfo{journal}{Journal of Machine Learning Research :
  JMLR}\/},  {\it \bibinfo{volume}{15 1}\/}, \bibinfo{pages}{1625--1651}.
\bibitem[{{Wallach} \& {Goffinet}(1989)}]{WALLACH1989299}
\bibinfo{author}{{Wallach}, D.}, \& \bibinfo{author}{{Goffinet}, B.}
  (\bibinfo{year}{1989}).
\newblock \bibinfo{title}{Mean squared error of prediction as a criterion for
  evaluating and comparing system models}.
\newblock {\it \bibinfo{journal}{Ecological Modelling}\/},  {\it
  \bibinfo{volume}{44}\/}, \bibinfo{pages}{299 -- 306}.
\bibitem[{Williams \& Barber(1998)}]{WilliamsBayesianGP98}
\bibinfo{author}{Williams, C.}, \& \bibinfo{author}{Barber, D.}
  (\bibinfo{year}{1998}).
\newblock \bibinfo{title}{Bayesian classification with {G}aussian processes}.
\newblock {\it \bibinfo{journal}{IEEE Transactions on Pattern Analysis and
  Machine Intelligence}\/},  {\it \bibinfo{volume}{20}\/},
  \bibinfo{pages}{1342--1351}. \DOIprefix\doi{10.1109/34.735807}.
\bibitem[{Zhang \& Wang(2010)}]{Zhang2010}
\bibinfo{author}{Zhang, H.}, \& \bibinfo{author}{Wang, Y.}
  (\bibinfo{year}{2010}).
\newblock \bibinfo{title}{Kriging and cross-validation for massive spatial
  data}.
\newblock {\it \bibinfo{journal}{Environmetrics}\/},  {\it
  \bibinfo{volume}{21}\/}, \bibinfo{pages}{290--304}.
  \DOIprefix\doi{https://doi.org/10.1002/env.1023}.
\bibitem[{Zhang et~al.(2020)Zhang, Zimmerman, Nettleton \&
  Nordman}]{Zhang_OOB_2018}
\bibinfo{author}{Zhang, H.}, \bibinfo{author}{Zimmerman, J.},
  \bibinfo{author}{Nettleton, D.}, \& \bibinfo{author}{Nordman, D.}
  (\bibinfo{year}{2020}).
\newblock \bibinfo{title}{Random forest prediction intervals}.
\newblock {\it \bibinfo{journal}{The American Statistician}\/},  {\it
  \bibinfo{volume}{74}\/}, \bibinfo{pages}{392 -- 406}.
\bibitem[{{Zhao}(1997)}]{ZHAO1997}
\bibinfo{author}{{Zhao}, J.} (\bibinfo{year}{1997}).
\newblock \bibinfo{title}{The lower semicontinuity of optimal solution sets}.
\newblock {\it \bibinfo{journal}{Journal of Mathematical Analysis and
  Applications}\/},  {\it \bibinfo{volume}{207}\/}, \bibinfo{pages}{240 --
  254}.
\bibitem[{{Zhou}(1998)}]{zhou1998adaptive}
\bibinfo{author}{{Zhou}, Y.} (\bibinfo{year}{1998}).
\newblock {\it \bibinfo{title}{Adaptive Importance Sampling for
  Integration}\/}.
\newblock Ph.D. thesis Stanford University.

\end{thebibliography}

\clearpage

\appendix

\section{Proofs of Propositions \ref{prop:prop1} - \ref{prop:prop3}.}
\label{appendix:A}

\subsection{Preliminary lemmas}

\begin{lemma}
\label{lemma1}
    Let $\mathbf{F}$ be a full rank matrix (hypothesis $\mathcal{H}_1$), let $\mathbf{K}$ be a positive definite matrix and let $\overline{\mathbf{K}}$ defined by $ \overline{\mathbf{K}} = \mathbf{K}^{-1} \left( \mathbf{I}_n - \mathbf{F}\left(\mathbf{F}^{\top} \mathbf{K}^{-1} \mathbf{F}\right)^{-1} \mathbf{F}^{\top} \mathbf{K}^{-1} \right)$ then {\rm Ker }$\overline{\mathbf{K}}$ = {\rm Im }$\mathbf{F}$ and $\overline{\mathbf{K}}$ is singular.
\end{lemma}
\begin{proof}
    Let $\overline{\mathbf{K}}$ be the matrix defined above. Suppose that $\vx \in$ {\rm Im }$\mathbf{F}$, then there exists $\vy$ such that $\vx = \mathbf{F} \vy$, and $\overline{\mathbf{K}} \vx = \mathbf{K}^{-1} \left( \mathbf{F} \vy - \mathbf{F}\left(\mathbf{F}^{\top} \mathbf{K}^{-1} \mathbf{F}\right)^{-1} \mathbf{F}^{\top} \mathbf{K}^{-1} \mathbf{F} \vy \right) = \mathbf{K}^{-1} \left( \mathbf{F} \vy  - \mathbf{F} \vy  \right) = \mathbf{0}$. Thus $\vx \in$ {\rm Ker }$\overline{\mathbf{K}}$.

    If $\vx \in$ {\rm Ker }$\overline{\mathbf{K}}$, then $\mathbf{K} \ \overline{\mathbf{K}}\vx = \mathbf{0}$, and $\vx = \mathbf{F}(\mathbf{F}^{\top} \mathbf{K}^{-1} \mathbf{F})^{-1}\mathbf{F}^{\top} \mathbf{K}^{-1} \vx = \mathbf{F} \vx' \in$ {\rm Im} $\mathbf{F}$.
        
    In case of Ordinary or Universal kriging, $p=\operatorname{rank}(\mathbf{F}) = \dim (${\rm Ker} $\overline{\mathbf{K}} ) \geq 1 $ which means that $\overline{\mathbf{K}}$ is not invertible.
\end{proof}

\begin{lemma}[ \cite{DeOliveira2007}]
\label{lemma2}
    Under the hypotheses of Lemma \ref{lemma1} and given the full rank regression matrix $\mathbf{F}$, there exists a matrix $\mathbf{W} \in \sR^{n \times (n-p)}$ satisfying :
    \begin{align}
        \mathbf{W}^{\top}\mathbf{W}&=\mathbf{I}_{n-p}, \\ \mathbf{F}^{\top}\mathbf{W}&=\mathbf{O}_{p \times (n-p)},
    \end{align}
    and
    \begin{equation}
        \overline{\mathbf{K}} = \mathbf{W} \left( \mathbf{W}^{\top} \mathbf{K} \mathbf{W} \right)^{-1} \mathbf{W}^{\top}.
    \end{equation}
\end{lemma}

\begin{lemma}
\label{lemma3}
    Under the hypotheses of Lemma \ref{lemma1}, if additionally hypothesis $\mathcal{H}_2$ holds true, then $\overline{\mathbf{K}}_{ii} > 0$ for all $i \in \{1,\ldots,n\}$.
\end{lemma}
    
\begin{proof}
    $\overline{\mathbf{K}}$ is a positive semi-definite matrix by Lemma \ref{lemma2} and we can write
    \begin{equation}
            \overline{\mathbf{K}} = \sum_{j=1}^n \lambda_j {\bf u}_j {\bf u}_j^{\top} ,
    \end{equation}
    with $\lambda_j\geq 0$ the eigenvalues of $\overline{\mathbf{K}}$ and $({\bf u}_j)_{j=1}^n$ the orthonormal basis of the eigenvectors. We have
    
    \begin{equation}
        \overline{\mathbf{K}}_{ii} = {\bf e}_i^{\top} \overline{\mathbf{K}} {\bf e}_i =  \sum_{j=1}^n \lambda_j ({\bf u}_j^{\top} {\bf e}_i)^2 .
    \end{equation}
    
    If $\overline{\mathbf{K}}_{ii}=0$, then ${\bf u}_j^{\top} {\bf e}_i =0$ for all $j$ such that $\lambda_j>0$. Therefore 
    \begin{equation}
        \overline{\mathbf{K}} {\bf e}_i = \sum_{j=1}^n \lambda_j ({\bf u}_j^{\top} {\bf e}_i ) {\bf u}_j = {\bf 0} ,
    \end{equation}

    which shows that ${\bf e}_i \in$ {\rm Ker }$\overline{\mathbf{K}}$, that is, ${\bf e}_i \in$ {\rm Im }$\mathbf{F}$ by Lemma \ref{lemma1}.
\end{proof}

\begin{lemma}
\label{lemma4}
    Let $\boldsymbol{\Pi} = \mathbf{W} \mathbf{W}^{\top} = \mathbf{I}_n - \mathbf{F} \left( \mathbf{F}^{\top} \mathbf{F} \right)^{-1} \mathbf{F}^{\top}$ the orthogonal projection matrix on {\rm (Im }$\mathbf{F})^{\perp}$ then, with the hypothesis $\mathcal{H}_2$,  $\left(\boldsymbol{\Pi}\right)_{i,i}\neq 0$ for all $i \in \{1,\ldots,n\}$.
\end{lemma}    
\begin{proof}
    This lemma is a direct application of Lemma \ref{lemma3} by choosing $\mathbf{K}=\mathbf{I}_n$.
\end{proof}

\subsection{Proof of Proposition \ref{prop:prop1}}

From preliminary lemmas, we show now the stronger result (stronger than Proposition \ref{prop:prop1}):    
\begin{lemma}
\label{lemma5}
   Under the hypotheses $\mathcal{H}_1-\mathcal{H}_3$, for any $\vtheta \in (0,+\infty)^d$, there exists $\sigma^2 \in (0,+\infty)$ such that $(\sigma^2, \vtheta) \in \mathcal{A}_{a, \delta}$.
\end{lemma}   

\begin{proof}
    Here $\sigma^2_{\epsilon} > 0$. Let us assume that $a > {1}/{2}$ (i.e. $q_a > 0$), then for $\vtheta$ fixed in $(0,+\infty)^d$, the limit of $\overline{\mathbf{K}}$ when $\sigma^2 \rightarrow 0$ is well defined and is equal to
    \begin{equation}
        \lim\limits_{\sigma^2 \rightarrow 0} \overline{\mathbf{K}} = \sigma^{-2}_{\epsilon} \ \mathbf{W} \mathbf{W}^{\top} = \sigma^{-2}_{\epsilon} \ \boldsymbol{\Pi}.
    \end{equation}

    By the hypothesis $\mathcal{H}_2$ and from Lemma \ref{lemma4} we can write for all $i \in \{1,\ldots,n\}$
    \begin{equation}
        \frac{ \left( \overline{\mathbf{K}} \vy \right)_i }{ \sqrt{ \left( \overline{\mathbf{K}} \right)_{i,i}}} \stackrel{\sigma^2 \rightarrow 0}{\longrightarrow} \frac{1}{\sigma_{\epsilon}} \frac{\left(\boldsymbol{\Pi} \vy \right)_i}{\sqrt{\left(\boldsymbol{\Pi} \right)_{i,i}}} .
    \end{equation}

    Since $h^+_{\delta} \leq h$ for all $\delta > 0$, then 
    \begin{equation}
        \lim\limits_{\sigma^2 \rightarrow 0} \psi^{\left(\delta \right)}_a(\sigma^2, \vtheta) \leq \lim\limits_{\sigma^2 \rightarrow 0} \psi_a(\sigma^2, \vtheta) = \frac{1}{n} \sum_{i=1}^{n} h \left( q_a - \frac{1}{\sigma_{\epsilon}} \frac{\left(\boldsymbol{\Pi} \vy \right)_i}{\sqrt{\left(\boldsymbol{\Pi} \right)_{i,i}}} \right) = \frac{k_{\epsilon}}{n} 
    \end{equation}

    When $\sigma^2 \rightarrow +\infty$, we have
    \begin{equation}
        \overline{\mathbf{K}} \stackrel{\sigma^2 \rightarrow +\infty}{\sim} \sigma^{-2} \ \overline{\mathbf{R}}_{\vtheta} ,
    \end{equation}
    where 
    \begin{equation}
        \overline{\mathbf{R}}_{\vtheta} = \mathbf{W} \left( \mathbf{W}^{\top} \mathbf{R}_{\vtheta} \mathbf{W} \right)^{-1} \mathbf{W}^{\top} .
    \end{equation}

    By lemma \ref{lemma3}, we have  ${\left(\overline{\mathbf{R}}_{\vtheta} \right)_{i,i}} > 0$ for all $i \in \{1,\ldots,n\}$ and we obtain that
    \begin{equation}
        \frac{1}{\sigma} \frac{\left(\overline{\mathbf{R}}_{\vtheta} \vy \right)_i}{\sqrt{\left(\overline{\mathbf{R}}_{\vtheta} \right)_{i,i}}} \stackrel{\sigma^2 \rightarrow +\infty}{\longrightarrow} 0.
    \end{equation}

    With $\delta$ small enough satisfying $\delta < q_a$, we obtain
    \begin{equation}
        \psi^{(\delta)}_a(\sigma^2, \vtheta) \stackrel{\sigma^2 \rightarrow +\infty}{\longrightarrow} \frac{1}{n} \sum_{i=1}^{n} h^+_{\delta} \left( q_a \right) = 1.
    \end{equation}

    Since $k_{\epsilon} < a n < n$ by hypothesis $\mathcal{H}_3$ and since $\psi^{(\delta)}_a$ is continuous, the Intermediate Value Theorem gives the existence of $\sigma_{\delta}^2 \in (0,+\infty)$ such that 
    \begin{equation}
        \psi^{(\delta)}_a(\sigma_{\delta}^2, \vtheta) = a , 
    \end{equation}
    which gives the desired result.
    
    Similarly, if $a < {a}/{2}$ then $q_a < 0 $ and
    \begin{equation}
        \lim\limits_{\sigma^2 \rightarrow 0} \psi^{\left(\delta\right)}_a(\sigma^2, \vtheta) \geq \lim\limits_{\sigma^2 \rightarrow 0} \psi_a(\sigma^2, \vtheta) = \frac{1}{n} \sum_{i=1}^{n} h \left( q_a - \frac{1}{\sigma_{\epsilon}} \frac{\left(\boldsymbol{\Pi} \vy \right)_i}{\sqrt{\left(\boldsymbol{\Pi} \right)_{i,i}}} \right) = \frac{k_{\epsilon}}{n} > a .
    \end{equation}

    When $\delta < q_{1-a}$, one obtains 
    \begin{equation}
        \psi^{(\delta)}_a(\sigma^2, \vtheta) \stackrel{\sigma^2 \rightarrow +\infty}{\longrightarrow} \frac{1}{n} \sum_{i=1}^{n} h^-_{\delta} \left( q_a \right) = 0 .
    \end{equation}
    
    By the hypothesis $\mathcal{H}_3$, one has the existence of $\sigma_{\delta}^2 \in (0,+\infty)$ such that 
    \begin{equation}
        \psi^{(\delta)}_a(\sigma_{\delta}^2, \vtheta) = a , 
    \end{equation}
    which completes the proof of the lemma.
\end{proof}

\subsection{Proof of Proposition \ref{prop:prop2}}

The existence of $\sigma^2_{\rm{opt}}(\lambda)$ for all $\lambda \in (0,+\infty)$ results directly from the following lemma \ref{lemma6} :

\begin{lemma}
\label{lemma6}
    For all $\lambda \in (0,+\infty)$, $H_{\delta}(\lambda)$ is a non-empty and compact subset of $\sR^+$ i.e. $H_{\delta}$ is compact-valued.
\end{lemma}

\begin{proof}
    By Lemma \ref{lemma5}, $H_{\delta}(\lambda)$ is non-empty for all $\lambda \in (0,+\infty)$.
    
    $H_{\delta}(\lambda)$ is closed since the functions $h^+_{\delta}, h^-_{\delta}$ are continuous and the map $(\sigma^2, \vtheta) \mapsto \overline{\mathbf{K}}$ is also continuous for all $ (\sigma^2, \vtheta)$ by the continuity of the kernel function $\vk^{\nu}_{., .}(\vx,\vx')$ for any $\nu > 0$ and $\vx,\vx' \in \mathcal{D}$.
    
    We now prove that $H_{\delta}(\lambda)$ is bounded. Let us assume that $a \in (1/2,1)$. If $H_{\delta}(\lambda)$ is not bounded then there exists a sequence $\left(\sigma^2_m\right)_{m \in \sN}$ of $H_{\delta}(\lambda)$ such that $\lim\limits_{m \rightarrow +\infty} \sigma_m^2=+\infty$ and, by continuity of $\psi^{(\delta)}_{a}$
    \begin{equation}
        a = \lim\limits_{m \rightarrow +\infty} \psi^{(\delta)}_{a}(\sigma_m^2, \lambda \vtheta_0) = \frac{1}{n} \sum_{i=1}^{n} h^+_{\delta} \left( q_a \right) = 1 ,
    \end{equation}
    which is a contradiction. Therefore, $H_{\delta}(\lambda)$ is closed and bounded, $H_{\delta}(\lambda)$ is compact.
\end{proof}

$\sigma^2_{\rm{opt}}(\lambda)$ can be seen the solution of a constrained maximization problem

\begin{equation}
    \sigma^2_{\rm{opt}}(\lambda) = - \max_{\sigma^2 \in H_{\delta}(\lambda)} u(\sigma^2, \lambda), \quad \lambda \in (0,+\infty) ,
\end{equation}

where $u(\sigma^2, \lambda) = -\sigma^2$ is a continuous function. $H_{\delta}$ is non-empty-valued and compact-valued by Lemma \ref{lemma6}, upper semi-continuous since $\psi^{(\delta)}_{a}$ is continuous on $[0,+\infty) \times (0,+\infty)^d$, and continuous if the hypothesis $\mathcal{H}_4$ is satisfied, the Maximum theorem (\citet{Berge1963}, p. 116) provides the continuity of $\sigma^2_{\rm{opt}}$ on $(0,+\infty)$.

\subsection{Proof of Proposition \ref{prop:prop3}}

Let $\vtheta_0$ be a solution of one of the problems described in (\ref{eq:OptimML}) or (\ref{eq:MSEsolution}). The continuity of $\mathcal{L}$ on $(0, +\infty)$ follows from the continuity of the trace function $\Tr(.)$, the continuity of the map $ (\sigma^2, \vtheta) \mapsto \overline{\mathbf{K}}$ and the continuity of $\sigma^2_{\rm{opt}}$ by proposition \ref{prop:prop2}.
    
Assume that $\lim\limits_{\lambda \rightarrow +\infty} \sigma^2_{\rm{opt}}(\lambda) \neq +\infty$, then there exists $M > 0$ such that for all $\lambda>0$ there exists $\lambda' \geq \lambda$ and $\sigma^2_{\rm{opt}}(\lambda') \leq M$. Hence, we can recursively build a sequence $\left( \lambda_m \right)_{m \in \sN}$ of integers such that $\lambda_{m+1} \geq \lambda_{m}+1$ and $\sigma^2_{\rm{opt}}(\lambda_m) \leq M$ for all $m \in \sN$.

By the Bolzano-Weierstrass theorem, we extract a convergent sub-sequence $\left( \lambda_{\phi(m)} \right)_{m \in \sN}$ where $\phi:\sN \to \sN$ such that $\sigma^2_{\rm{opt}}(\lambda_{\phi(m)}) \stackrel{m \rightarrow +\infty}{\longrightarrow} \sigma^2_{\infty} < +\infty$ and

\begin{equation}
    \mathbf{K}_{\sigma^2_{\rm{opt}}(\lambda_{\phi(m)}), \lambda_{\phi(m)} \vtheta_0} \stackrel{m \rightarrow +\infty}{\longrightarrow} \sigma^2_{\infty} \mathbf{J} + \sigma^2_{\epsilon} \mathbf{I}_n = \mathbf{K}_{\infty} .
\end{equation}

When there is a nugget effect $\sigma^2_{\epsilon} > 0$, the limit of $\overline{\mathbf{K}}_m := \overline{\mathbf{K}}_{\sigma^2_{\rm{opt}}(\lambda_{\phi(m)}), \lambda_{\phi(m)} \vtheta_0}$ when $m\rightarrow +\infty$ exists  because the matrix $\mathbf{K}_{\infty}$ is nonsingular by the auxiliary fact 1 of \cite{Berger2001}
\begin{equation}
    \operatorname{det}\mathbf{K}_{\infty} = \left( \frac{\sigma^2_{\epsilon}}{\sigma^2_{\infty}} \right)^{n} \left( 1 +\frac{\sigma^2_{\epsilon}}{\sigma^2_{\infty}} \ {\bf e} ^{\top}  \mathbf{I}_n  {\bf e} \right) =  \left( \frac{\sigma^2_{\epsilon}}{\sigma^2_{\infty}} \right)^{n} \left( 1 + n \frac{\sigma^2_{\epsilon}}{\sigma^2_{\infty}} \right) > 0.
\end{equation}

From hypothesis $\mathcal{H}_1$, ${\bf e}$ is a column of ${\bf F}$ and we can prove that 
\begin{equation}
    \begin{aligned}\overline{\mathbf{K}}_m \stackrel{m \rightarrow +\infty}{\longrightarrow} \overline{\mathbf{K}}_{\infty} :&= \mathbf{W} \left( \mathbf{W}^{\top}  \left( \sigma^2_{\infty} \mathbf{J} + \sigma^2_{\epsilon} \mathbf{I}_n \right)\mathbf{W} \right)^{-1} \mathbf{W}^{\top} \\ 
    &= \sigma^{-2}_{\epsilon} \ \mathbf{W} \left( \mathbf{W}^{\top} \mathbf{W} \right)^{-1} \mathbf{W}^{\top} = \sigma^{-2}_{\epsilon} \ \boldsymbol{\Pi}. \end{aligned} 
\end{equation}

By hypothesis $\mathcal{H}_2$, the Leave-One-Out formulas (\ref{eq:LOO_mean}-\ref{eq:LOO_var}) give for all $i \in \{1,\ldots, n\}$ 
\begin{equation}
    \frac{ \left( \overline{\mathbf{K}}_m \vy \right)_i }{ \sqrt{ \left( \overline{\mathbf{K}}_m \right)_{i,i}}} \stackrel{m \rightarrow +\infty}{\longrightarrow} \frac{1}{\sigma_{\epsilon}} \frac{ \left( \boldsymbol{\Pi} \vy \right)_i }{\sqrt{ \left( \boldsymbol{\Pi} \right)_{i,i}}} .
\end{equation}

If $a>1/2$ for example and by definition of $\sigma^2_{\rm{opt}}(\lambda_{\phi(m)})$, one obtains
\begin{equation}
    \begin{aligned} a &= \frac{1}{n} \sum_{i=1}^{n} h^+_{\delta} \left( q_a - \frac{\left(\overline{\mathbf{K}}_m \vy \right)_i}{\sqrt{\left(\overline{\mathbf{K}}_m \right)_{i,i}}} \right) \\
    & \stackrel{m\rightarrow +\infty}{\longrightarrow} \frac{1}{n} \sum_{i=1}^{n} h^+_{\delta} \left( q_a -  \frac{\left(\overline{\mathbf{K}}_{\infty} \vy \right)_i}{\sqrt{\left(\overline{\mathbf{K}}_{\infty} \right)_{i,i}}} \right) \\
    &= \frac{1}{n} \sum_{i=1}^{n} h^+_{\delta} \left( q_a - \frac{1}{\sigma_{\epsilon}} \frac{\left(\boldsymbol{\Pi} \vy \right)_i}{\sqrt{\left(\boldsymbol{\Pi} \right)_{i,i}}} \right) = \frac{ k_{\epsilon}}{n} < a, \end{aligned}
\end{equation}

which is contradictory. Therefore, $\lim\limits_{\lambda \rightarrow +\infty} \sigma^2_{\rm{opt}}(\lambda) = +\infty$ and $\mathcal{L}$ is coercive. The case $a<1/2$ can be addressed in the same way.

\section{The no-nugget case.}
\label{appendix:B}

\subsection{Proof of the existence of a solution to Problem (\ref{Eq:OptimLOO})}

In the absence of $\sigma^2_{\epsilon} = 0$, it follows from the Leave-One-Out formulas that,  for all $i \in \{1,\ldots,n\}$ 
\begin{equation}
    \frac{ \left( \overline{\mathbf{K}} \vy \right)_i }{ \sqrt{ \left( \overline{\mathbf{K}} \right)_{i,i}}} = \frac{1}{\sigma} \frac{\left(\overline{\mathbf{R}}_{\vtheta} \vy \right)_i}{\sqrt{\left(\overline{\mathbf{R}}_{\vtheta} \right)_{i,i}}},
\end{equation}
which is a monotonic function in $\sigma^2$ when $\vtheta$ is fixed in $(0,\infty)^d$.

Let $\vtheta$ be fixed in $(0,+\infty)^d$ and let $a>1/2$. The proportion $\psi^{(\delta)}_a(\sigma^2,\vtheta)$ has the limit 
\begin{equation}
    \lim\limits_{\sigma^2 \rightarrow +\infty}  \psi^{(\delta)}_a(\sigma^2, \vtheta) = \frac{1}{n} \sum_{i=1}^n h^+_{\delta}(q_a) = 1,
\end{equation}

and, if $\sigma^2 \rightarrow 0$, it has the limit
\begin{equation}
    \begin{aligned}\lim\limits_{\sigma^2 \rightarrow 0}  \psi^{(\delta)}_a(\sigma^2,\vtheta) &= \frac{1}{n} \operatorname{Card} \left\{i \in \{1,\ldots,n\} ,  \left( \overline{\mathbf{R}}_{\vtheta} \vy \right)_i \leq 0 \right\} = \frac{k_{\vtheta}}{n} .
    \end{aligned}
\end{equation}

Let $\theta $ denote the norm of $\vtheta$ (i.e. $\theta= \| \vtheta \|$) and consider the set $\mathcal{J}=\left\{i \in \{1,\ldots,n\} ,  \left( \boldsymbol{\Pi} \vy \right)_i \leq 0 \right\}$. For $i \in \mathcal{J}^c$, one has $\left( \boldsymbol{\Pi} \vy \right)_i > 0 $, and, since $\overline{\mathbf{R}}_{\vtheta}$ converges to  $\boldsymbol{\Pi}$ when $\theta \rightarrow 0$
\begin{equation}
    \forall i \in \mathcal{J}^c \ : \left( \overline{\mathbf{R}}_{\vtheta} \vy \right)_i \stackrel{\theta \rightarrow 0}{\longrightarrow} \left( \boldsymbol{\Pi} \vy \right)_i > 0 
\end{equation}

It results that, there exists $\theta_c>0$ such that if $\vtheta \in \mathcal{B}_r(\boldsymbol{0}, \theta_c)$ (the open ball of radius $\theta_c$ centered at $\boldsymbol{0}$) then $ \left( \overline{\mathbf{R}}_{\vtheta} \vy \right)_i>0$ for any $i \in \mathcal{J}^c$. Consequently, one gets for any $\vtheta \in \mathcal{B}_r(\boldsymbol{0}, \theta_c)$
\begin{equation}
    \operatorname{Card} \left\{i \in \{1,\ldots,n\} ,  \left( \overline{\mathbf{R}}_{\vtheta} \vy \right)_i > 0 \right\} \geq \operatorname{Card} (\mathcal{J}^c) = n-k_{\epsilon}.
\end{equation}

Hence
\begin{equation}
    k_{\vtheta} = \operatorname{Card} \left\{i \in \{1,\ldots,n\} ,  \left( \overline{\mathbf{R}}_{\vtheta} \vy \right)_i \leq 0 \right\} \leq k_{\epsilon}.
\end{equation}

Therefore, if $\vtheta$ belongs to a neighborhood of $\boldsymbol{0}$, the condition $k_{\vtheta} \leq k_{\epsilon}$ is satisfied and, under the hypothesis $\mathcal{H}_3$, the set of solutions $\mathcal{A}_{a,\delta}$ is also non-empty.

\subsection{Proof of the Coercivity}

Let assume that, under some conditions on $\vy$, $\lambda \mapsto \sigma^2_{\rm{opt}}(\lambda)$ is well-defined for all $\lambda \in (0,+\infty)$. In the absence of nugget effect $\sigma^2_{\epsilon}=0$, the limit of $ \overline{\mathbf{R}}_{\lambda \vtheta_0}$ does not exist when $\lambda \rightarrow +\infty$. Still, we can assume that the correlation matrix $\mathbf{R}_{\lambda \vtheta_0}$ satisfies \citep{Berger2001}
\begin{equation}
    \label{eq:assumpt_Rtheta}
    \mathbf{R}_{\lambda \vtheta_0}= {\bf J} + g_{\lambda} \left( {\bf D}_0 + o(1) \right) ,
\end{equation} 
where
\begin{itemize}
    \item[--] $\lambda \mapsto g_{\lambda}$ is a continuous function such that $\lim\limits_{\lambda \rightarrow +\infty} g_{\lambda}=0$.
    \item[--] ${\bf D}_0$ and ${\bf J}={\bf e} {\bf e}^{\top}$ are fixed symmetric matrices.
\end{itemize}
    
${\bf D}_0$ can be singular or nonsingular depending on the chosen kernel $\vk$. A review of Yagloom's book \citep{Yaglom1989} shows that ${\bf D}_0$ is nonsingular only for Power-Exponential ($q<2$) and Mat\'ern kernels with smoothness parameter $\nu < 1$ like the Exponential kernel ($\nu = 1/2$ in (\ref{eq:maternclass})). For the rest of Mat\'ern kernels with smoothness parameter $\nu \geq 1$ ${\bf D}_0$ becomes singular. 

\paragraph{Case 1 : ${\bf D}_0$ is nonsingular } \hfill
\label{par:case_1}
    
In this case, let $\mathbf{D}_{\lambda}=g_{\lambda} \ \mathbf{D}_0 \left(1+o(1)\right)$ such that
\begin{equation}
    \mathbf{R}_{\lambda \vtheta_0}= {\bf J} + \mathbf{D}_{\lambda} .
\end{equation}
    
We consider the matrix $ \overline{\bf R}_{\lambda \vtheta_0}$ in $\overline{\mathbf{K}}=\sigma^{-2} \ \overline{\mathbf{R}}_{\lambda \vtheta_0}$, we have
\begin{equation}
    \overline{\bf R}_{\lambda \vtheta_0}  = \mathbf{R}_{\lambda \vtheta_0}^{-1}\left[ \mathbf{I}_n- \mathbf{F}\left(\mathbf{F}^{\top} \mathbf{R}_{\lambda \vtheta_0}^{-1} \mathbf{F}\right)^{-1} \mathbf{F}^{\top} \mathbf{R}_{\lambda \vtheta_0}^{-1} \right] .
\end{equation}
    
By using Lemma 4, Appendix B3 in \cite{Berger2001} and under assumption that ${\bf e \in}$ {\rm Im }$\mathbf{F}$ (hypothesis $\mathcal{H}_1$), we have
\begin{equation}
    \overline{\bf R}_{\lambda \vtheta_0} = {\bf D}_\lambda^{-1} \left[ \mathbf{I}_n - \mathbf{F} \left(\mathbf{F}^{\top} {\bf D}_\lambda^{-1} \mathbf{F}\right)^{-1}\mathbf{F}^{\top} {\bf D}_\lambda^{-1} \right] .
\end{equation}
Then we get
\begin{equation}
    \overline{\bf R}_{\lambda \vtheta_0}  = g_\lambda^{-1} \left[ {\bf D}_0^{-1} \left( \mathbf{I}_n - \mathbf{F} \left(\mathbf{F}^{\top} {\bf D}_0^{-1} \mathbf{F}\right)^{-1}\mathbf{F}^{\top} {\bf D}_0^{-1} \right) + o(1) \right] .
\end{equation}
    
Finally
\begin{equation}
\label{1stA}
    \overline{\mathbf{R}}_{\lambda \vtheta_0} \stackrel{\lambda \rightarrow +\infty}{\sim} g_\lambda^{-1} \mathbf{A} ,
\end{equation}
where 
\begin{equation}
\label{defA1} 
    \mathbf{A} = \mathbf{D}_0^{-1} \left( \mathbf{I}_n - \mathbf{F}\left(\mathbf{F}^{\top} \mathbf{D}_0^{-1} \mathbf{F}\right)^{-1} \mathbf{F}^{\top} \mathbf{D}_0^{-1} \right).
\end{equation}

{\it Hypothesis $\mathcal{H}_5$~: Let $\mathbf{A}$ be the matrix defined in (\ref{defA1}). We assume that $\vy$ does not belong to a family of vectors such that $(\mathbf{A} \vy)_i = 0$ for all $i \in \{1,\ldots,n\}$ and that $\operatorname{Card} \left\{i \in \{1,\ldots,n\} ,  \left( \mathbf{A} \vy \right)_i \leq 0 \right\} \neq na$.} 

By applying Lemmas \ref{lemma1} and \ref{lemma2} on $\mathbf{D}_0$, we show that $\left( \mathbf{A} \right)_{i i}\neq0$ and we can write for all $i$ in $\{1,...,n\}$
\begin{equation}
    \frac{ \left( \overline{\mathbf{R}}_{\lambda \vtheta_0} \vy \right)_i }{ \sqrt{ \left( \overline{\mathbf{R}}_{\lambda \vtheta_0} \right)_{i i}}} \stackrel{\lambda \rightarrow +\infty}{\sim} g_\lambda^{-1/2} \frac{ \left( \mathbf{A} \vy \right)_i }{ \sqrt{ \left( \mathbf{A} \right)_{i i}}} .
\end{equation}
    
Analogously to the proof of Proposition \ref{prop:prop3}, if we assume that $\lim\limits_{\lambda \rightarrow +\infty} \sigma^2_{\rm{opt}}(\lambda) \neq +\infty$ and by taking a sub-sequence $\left( \sigma^2_{\rm{opt}}(\lambda_{\psi(m)}) \right)_{m \in \sN}$ converging to ${\sigma^2_\infty}$
\begin{equation}
    \frac{1}{\sigma_\infty} \ g_{\lambda_{\psi(m)}}^{-1/2} \frac{ \left( \mathbf{A} \vy \right)_i }{ \sqrt{ \left( \mathbf{A} \right)_{i i}}} \stackrel{m \rightarrow +\infty}{\longrightarrow} \left\{ \begin{array}{l l}
    +\infty & \quad \text{if $ \left( \mathbf{A} \vy \right)_i > 0$}\\
    -\infty & \quad \text{otherwise}\\ \end{array} \right.
\end{equation}
    
The limit $\psi^{(\delta)}_a(\sigma_{\rm{opt}}^2(\lambda_{\psi(m)}), \lambda_{\psi(m)} \vtheta_0)$  when $m \rightarrow +\infty$ exists and is equal to
\begin{equation}
\label{eq:quasiGaussLim}
    a = \lim\limits_{m \rightarrow +\infty} \psi^{(\delta)}_a(\sigma_{\rm{opt}}^2(\lambda_{\psi(m)}), \lambda_{\psi(m)} \vtheta_0) = \frac{1}{n} \operatorname{Card} \left\{i \in \{1,\ldots,n\} ,  \left( \mathbf{A} \vy \right)_i \leq 0 \right\} ,
\end{equation}
which is contradictory and completes the proof.

\paragraph{Case 2 : $\mathbf{D}_0$ is singular } \hfill
    
In this case, one needs to go further in the Taylor expansion of $\overline{\mathbf{R}}_{\lambda \vtheta_0}$. We consider the matrix $\mathbf{W}$ in Lemma \ref{lemma3}, by Lemma 6 of \cite{Ren2012ObjectiveBA} 
\begin{equation}
    \overline{\mathbf{R}}_{\lambda \vtheta_0} = \mathbf{W} \left( \mathbf{W}^{\top} \mathbf{R}_{\lambda \vtheta_0} \mathbf{W} \right)^{-1} \mathbf{W}^{\top} .
\end{equation}

By setting $\mathbf{\Sigma}_{\lambda} = \mathbf{W}^{\top} \mathbf{R}_{\lambda \vtheta_0} \mathbf{W}$, the asymptotic study of $\overline{\mathbf{R}}_{\lambda \vtheta_0}$ is equivalent to the asymptotic study of $\mathbf{\Sigma}_{\lambda}$. In case of Mat\'ern kernel with noninteger smoothness $\nu \geq 1$, the matrix $\mathbf{\Sigma}_{\lambda}$ can be written as \citep{mure2020propriety}
    
\begin{equation}
\label{eq:Sigma_theta}
    \mathbf{\Sigma}_{\lambda} = g_\lambda \left(\mathbf{W}^{\top} \mathbf{D}_1 \mathbf{W}+g^{*}_\lambda \mathbf{W}^{\top} \mathbf{D}_1^{*} \mathbf{W}+\mathbf{R}_{g}(\lambda)\right) ,
\end{equation}
where 
\begin{itemize}
    \item[--] Either $g_\lambda = c \lambda^{−2k_1}$ with $k_1$ a nonnegative integer, or $g_\lambda = c \lambda^{−2\nu}$.
    \item[--] $g^*_\lambda = c^* \lambda^{−2l}$ with $l \in (0, +\infty)$ .
    \item[--] $\mathbf{R}_g$ is a differentiable mapping from $[0, +\infty)$ to $\mathcal{M}_n$ such that $\left\Vert \mathbf{R}_g \left(\lambda\right) \right\Vert= o(\lambda^{−2l})$.
    \item[--] $\mathbf{D}_1$ and $\mathbf{D}_1^*$ are both fixed symmetric matrices with elements $\|x_i - x_j \|^{2k}$ where $k \in {k_1} \cup {\nu}$ for $\mathbf{D}_1$ and $k = l$ for $\mathbf{D}^*_1$.
\end{itemize}
    
The matrix $\mathbf{W}^{\top} \mathbf{D}_1 \mathbf{W}+g^{*}_\lambda \mathbf{W}^{\top} \mathbf{D}_1^{*} \mathbf{W}$ is nonsingular when $\lambda \rightarrow +\infty$, whether if $\mathbf{W}^{\top} \mathbf{D}_1 \mathbf{W}$ is nonsingular or if it is singular. 

The case where $\mathbf{W}^{\top} \mathbf{D}_1 \mathbf{W}$ is nonsingular happens for Mat\'ern kernels with smoothness $1 \leq \nu < 2$ \citep{mure2020propriety}, whereas the other case occurs for regular and smooth Mat\'ern kernels with $ \nu \geq 2$. These kernels are however less robust in uncertainty quantification so we will give only the proof for less smooth kernels with $1 \leq \nu < 2$ in particular the  Mat\'ern 3/2 kernel.

In this case, we write $\mathbf{\Sigma}_{\lambda}$ in (\ref{eq:Sigma_theta}) as
\begin{equation}
    \mathbf{\Sigma}_{\lambda} = g_\lambda \mathbf{W}^{\top} \mathbf{D}_1 \mathbf{W} \left(\mathbf{I}_n+g^*_\lambda \left(\mathbf{W}^{\top} \mathbf{D}_1 \mathbf{W} \right)^{-1}\left( \mathbf{W}^{\top} \mathbf{D}_1^{*} \mathbf{W}+\mathbf{R}_{g}(\lambda) \right) \right) .
\end{equation}
    
As $\mathbf{W}$ is full rank matrix, $\mathbf{\Sigma}_{\lambda}$ is non-singular and
\begin{equation}
\label{eq:inverseSigma}
    \mathbf{\Sigma}_{\lambda}^{-1} = g^{-1}_\lambda \left(\mathbf{I}_n+g^*_\lambda \left(\mathbf{W}^{\top} \mathbf{D}_1 \mathbf{W} \right)^{-1}\left( \mathbf{W}^{\top} \mathbf{D}_1^{*} \mathbf{W}+\mathbf{R}_{g}(\lambda) \right) \right)^{-1} \left(\mathbf{W}^{\top} \mathbf{D}_1 \mathbf{W} \right)^{-1} .
\end{equation}
         
Let $\mathbf{M}_{\lambda} = g^*_\lambda \left(\mathbf{W}^{\top} \mathbf{D}_1 \mathbf{W} \right)^{-1}\left( \mathbf{W}^{\top} \mathbf{D}_1^{*} \mathbf{W}+\mathbf{R}_{g}(\lambda) \right) $, since $ \| \mathbf{M}_{\lambda} \| \stackrel{\lambda \rightarrow +\infty}{\longrightarrow} 0 $, we can assume that $\| \mathbf{M}_{\lambda} \| <1 $ when $\lambda$ is large enough and apply the Taylor series expansion at order 1
\begin{equation}
    \begin{aligned} \left[\mathbf{I}_n+g^*_\lambda \left(\mathbf{W}^{\top} \mathbf{D}_1 \mathbf{W} \right)^{-1} \right. &  \left. \left( \mathbf{W}^{\top} \mathbf{D}_1^{*} \mathbf{W}+\mathbf{R}_{g}(\lambda) \right) \right]^{-1} = \mathbf{I}_n-g^*_\lambda \left(\mathbf{W}^{\top} \mathbf{D}_1 \mathbf{W} \right)^{-1} \\ & \times \left( \mathbf{W}^{\top} \mathbf{D}_1^{*} \mathbf{W}+\mathbf{R}_{g}(\lambda) + o(g^*_\lambda) \right) . \end{aligned}
\end{equation}

Then, we plug this quantity into the equation  (\ref{eq:inverseSigma})
\begin{equation}
    \begin{aligned} \mathbf{\Sigma}_{\lambda}^{-1} &= g^{-1}_\lambda \left( \mathbf{I}_n-g^*_\lambda \left(\mathbf{W}^{\top} \mathbf{D}_1 \mathbf{W} \right)^{-1}\left( \mathbf{W}^{\top} \mathbf{D}_1^{*} \mathbf{W}+\mathbf{R}_{g}(\lambda) \right) \right) \left(\mathbf{W}^{\top} \mathbf{D}_1 \mathbf{W} \right)^{-1} \\
    &= g^{-1}_\lambda \left[ \left(\mathbf{W}^{\top} \mathbf{D}_1 \mathbf{W} \right)^{-1}-g^*_\lambda \left(\mathbf{W}^{\top} \mathbf{D}_1 \mathbf{W} \right)^{-1}\left( \mathbf{W}^{\top} \mathbf{D}_1^{*} \mathbf{W}+\mathbf{R}_{g}(\lambda) \right) \left(\mathbf{W}^{\top} \mathbf{D}_1 \mathbf{W} \right)^{-1} \right] . \end{aligned}
\end{equation}
         
Finally, we can write the matrix $\overline{\mathbf{R}}_{\lambda \vtheta_0}$ as
\begin{equation}
    \overline{\mathbf{R}}_{\lambda \vtheta_0} = g^{-1}_\lambda \mathbf{W} \left[ \left(\mathbf{W}^{\top} \mathbf{D}_1 \mathbf{W} \right)^{-1}-g^*_\lambda \left(\mathbf{W}^{\top} \mathbf{D}_1 \mathbf{W} \right)^{-1}\left( \mathbf{W}^{\top} \mathbf{D}_1^{*} \mathbf{W}+\mathbf{R}_{g}(\lambda) \right) \left(\mathbf{W}^{\top} \mathbf{D}_1 \mathbf{W} \right)^{-1} \right] \mathbf{W}^{\top} .
\end{equation}
        
We can also simply the previous expression into
\begin{equation}
\label{2ndA}
    \overline{\mathbf{R}}_{\lambda \vtheta_0} = g^{-1}_\lambda \left( \mathbf{A} - \mathbf{B}_{\lambda} \right) ,
\end{equation}

where {$\mathbf{A}$ is a fixed matrix and $\mathbf{B}_{\lambda} \stackrel{\lambda \rightarrow +\infty}{=} o(1)$ such that}
\begin{flalign}
\label{defA2}
    &\mathbf{A} = \mathbf{W} \left(\mathbf{W}^{\top} \mathbf{D}_1 \mathbf{W} \right)^{-1}  \mathbf{W}^{\top} \\
    &\mathbf{B}_{\lambda} = g^*_\lambda \ \mathbf{W} \left(\mathbf{W}^{\top} \mathbf{D}_1 \mathbf{W} \right)^{-1}\left( \mathbf{W}^{\top} \mathbf{D}_1^{*} \mathbf{W}+\mathbf{R}_{g}(\lambda) \right) \left(\mathbf{W}^{\top} \mathbf{D}_1 \mathbf{W} \right)^{-1} \mathbf{W}^{\top}.
\end{flalign}       

Or, equivalently, 
\begin{equation}
    \overline{\mathbf{R}}_{\lambda \vtheta_0} \stackrel{\lambda \rightarrow +\infty}{\sim} g_\lambda^{-1} \mathbf{A}.
\end{equation} 

\begin{lemma}
\label{lemma7}
    Let $\mathbf{A}$ be the matrix defined in (\ref{defA2}) , then $\mathbf{A}_{i i} \neq 0$ for all $i \in \{1,\ldots,n\}$.
\end{lemma}
\begin{proof}
    $\mathbf{A}$ is non-singular because
    \begin{equation}
        \operatorname{det} \mathbf{A} = \operatorname{det} \mathbf{W} \left(\mathbf{W}^{\top} \mathbf{D}_1 \mathbf{W} \right)^{-1} \mathbf{W}^{\top} = \operatorname{det} \left(\mathbf{W}^{\top} \mathbf{D}_1 \mathbf{W} \right)^{-1} \neq 0.
    \end{equation}
    $\mathbf{A}$ is then a positive definite matrix
    \begin{equation}
            \mathbf{A}_{ii} = {\bf e}_i^{\top} \mathbf{A} {\bf e}_i > 0 .
    \end{equation}
    \end{proof}
{\it Hypothesis $\mathcal{H}_6$~: Let $\mathbf{A}$ be the matrix defined in (\ref{defA2}). We assume that $\vy$ does not belong to a family of vectors such that $(\mathbf{A} \vy)_i = 0$ for all $i \in \{1,\ldots,n\}$ and that $\operatorname{Card} \left\{i \in \{1,\ldots,n\} ,  \left( \mathbf{A} \vy \right)_i \leq 0 \right\} \neq na$.} 

With Lemma \ref{lemma6} and Hypothesis $\mathcal{H}_6$, the proof of the divergence of $\sigma_{\rm{opt}}^2(\lambda)$ when $\lambda \rightarrow +\infty$ is similar to the previous case when $\mathbf{D}_0$ is nonsingular.

\begin{remark}
    The hypotheses $\mathcal{H}_5$ and $\mathcal{H}_6$ are not restrictive, one can verify numerically, that each component of $\mathbf{A} \vy$ is not null where  $\mathbf{A}$ is one of the matrices defined in (\ref{defA1}) or (\ref{defA2}).

\end{remark}

\end{document}